\documentclass[runningheads]{llncs}

\usepackage[T1]{fontenc}
\usepackage{newtxtext}
\usepackage[varvw]{newtxmath}
\usepackage{graphicx}
\usepackage{amsmath}
\usepackage{amsfonts}
\usepackage{graphicx}
\usepackage{hyperref}
\usepackage{bookmark}
\bookmarksetup{
  numbered, 
  open,
  depth=3,
}
\usepackage{tikz-cd}
\usepackage{xcolor}
\usepackage{wasysym}
\usepackage{tikz}
\usepackage{prftree}
\usepackage{mathtools}
\usepackage{stmaryrd}
\usepackage{trimclip}
\usepackage{bbold}
\usepackage{wrapfig}
\usepackage{anyfontsize}
\usepackage{environ}
\usepackage[title,titletoc]{appendix}
\usepackage{etoolbox}
\usepackage[framemethod=default]{mdframed}
\usepackage{enumitem}

\usepackage[T1]{fontenc}
\usepackage{newtxtext}
\usepackage[varvw]{newtxmath}
\usepackage{graphicx}
\usepackage{color}

\usepackage[switch]{lineno}
\input{macros}

\setlist[description]{topsep=2pt,font=\itshape}

\allowdisplaybreaks[1]
\begin{document}

\ifdefined\pdfsuppressptexinfo \pdfsuppressptexinfo=-1 \fi

\title{A Direct Translation from Abstract Cyclic Proofs to Inductive Proofs}

\author{Lide Grotenhuis%\orcidID{TODO} 
     \and
     Dani\"el Otten\orcidID{0000-0003-2557-3959}%
}
\authorrunning{L. Grotenhuis and D. Otten}

\institute{ILLC, University of Amsterdam, Amsterdam, The Netherlands \\
\email{lidegrotenhuis@gmail.com} and \email{daniel@otten.co}
}

\maketitle

\begin{abstract} %The abstract should briefly summarize the contents of the paper in 150–250 words.
We identify the notion of a `well-organized' cyclic proof, show that any cyclic proof can be unfolded to a well-organized one, and present a constructive and relatively simple translation from well-organized cyclic proofs to ordinary well-founded proofs.
We work with a general class of cyclic proof systems, and translate from such a system to an induced inductive proof system: a multi-sorted predicate logic that extends the cyclic system with a well-founded induction principle.
This setting is broad enough to cover trace-based cyclic proofs and size-change terminating functions.
Our unfolding technique is a refinement of earlier annotation-based approaches, and allows us to handle cyclic proofs where the underlying inductive sorts need not be linearly ordered.
    %in which the soundness condition corresponds to the size-change termination principle of recursive functions
    
    % We identify a notion of well-organized cyclic proof, and present a fully constructive and relatively simple translation from well-organized cyclic proofs to ordinary finite proofs with an induction rule. This notion is defined in a general context of abstract cyclic proof systems, in which the soundness condition corresponds to the size-change termination principle. For the inductive proof system, we take a multi-sorted predicate logic that extends the cyclic system with first-order reasoning and a well-founded induction principle. Moreover, through a refined annotation technique, we show that every cyclic proof in our framework can be unfolded into one that is well-organized. This refined technique in particular allows us to handle cyclic proofs where the underlying inductive sort is non-linear.
    
\end{abstract}

\renewcommand\phi\varphi
\renewcommand\defeq\coloneq
\colorlet{green}{green!80!black}
\allowdisplaybreaks[1]

%%novalidate

\section{Introduction}

% introduce cyclic vs inductive proofs.
A key question in cyclic proof theory is how cyclic proofs compare to inductive proofs.
A \emph{cyclic proof} is a finite derivation with back-edges (representing an infinite derivation) that satisfies a specified \emph{soundness condition} on its infinite paths \cite{SprengerDamComp,simpson2017cyclic,marti2021focus,baelde2022bouncing,afshari2023proof}.
In contrast, an \emph{inductive proof} is a finite derivation that uses an explicit induction rule.
By allowing cycles, cyclic proof systems circumvent the need for such a rule, which simplifies proof search since we avoid having to `guess' an induction hypothesis.

% translating cyclic to inductive
\paragraph{From cycles to induction rules.} In general, showing that cyclic proofs subsume their inductive counterpart is straightforward, while the converse is hard and sometimes even fails \cite{berardi_classical_2019}.
Foundational equivalence results were proven by Sprenger \& Dam \cite{sprenger_structure_2003} for the first-order $\mu$-calculus and by Simpson \cite{simpson2017cyclic} and Berardi \& Tatsuta \cite{berardi_equivalence_2017,berardi2017equivalence} for Peano arithmetic and its extension with (a fixed signature of) inductive definitions.

% metamathematical
In the latter two works, the reduction from cyclic to inductive proofs relies on the formalization of mathematical results inside a suitable arithmetical theory.
This `metamathematical approach' was further developed by Curzi, Das, and Melgaard, resulting in refined equivalence results for (extensions of) arithmetic \cite{das2020logical,das_melgaard,curzim_elgaard2026cyclic}, and an equi-expressiveness result \cite{curzi2023computationalex} stating that cyclic and inductive proofs in intuitionistic logic with explicit fixpoints represent the same first-order functions on natural numbers. 

% direct 
In contrast, Sprenger \& Dam took a more direct approach. Starting from a cyclic proof with an `induction order' on its back-edges, they first showed that any such proof can be unfolded into one whose tree-structure is compatible with the induction order.
Secondly, they showed that a cyclic proof with a compatible tree-structure can be readily translated into an inductive proof by simply replacing each back-edge by an appropriate application of the induction rule.
This `unfolding approach' was later used by Leigh \& Wehr \cite{leigh2025unravelling} to reprove a version of Simpson's result, and was further refined in the PhD thesis of Wehr \cite{wehr2025cyclic}.
In particular, Wehr showed that every cyclic proof system with an $\omega$-regular soundness condition can be unfolded into one with an induction order.  

% curry-howard
\paragraph{Conservativity of recursive programs.}
Comparing cyclic and inductive proofs is relevant beyond proof theory.
Under the Curry-Howard correspondence, recursive programs provide another instance of cyclic proofs: we can view recursive calls as back-edges, and a termination criterion as a soundness condition. 
% \[\begin{array}{c|c}
%     \text{cyclic proof} & \text{recursive function} \\
%     \hline
%     \text{(co)inductive formula} & \text{(co)inductive sort} \\
%     \text{cycle} & \text{recursive call} \\
%     \text{soundness condition} & \text{termination condition}
% \end{array}\]
For example, proof assistants based on type theory (\texttt{Agda}, \texttt{Lean}, \texttt{Rocq}, \dots) allow the user to define functions on inductive types using pattern matching and recursive calls that satisfy a termination criterion.
Defining functions in this way is often significantly easier than defining a function using the primitive elimination rule (that is, the induction rule) of the underlying inductive types.

To understand the theory implemented by proof assistants, it is important to know whether the accepted recursive programs can already be defined using the primitive eliminators.
At the moment, the known conservativity results \cite{EliminatingDependentPatternMatching,EliminatingDependentPatternMatchingWithoutK} cover \emph{structural recursion}: a single function with an inductive input that is decreased in every recursive call. Meanwhile, proof assistants like \texttt{Agda}, \texttt{Dedukti}, and \texttt{Isabelle}/\texttt{HOL} already accept mutually recursive functions that satisfy the \emph{size-change termination} (SCT) principle: for any infinite sequence of recursive calls, there is eventually an input that progresses infinitely often \cite{lee_size-change_2001,DependentSizeChangeTermination,thibodeau_intensional_2020,hyvernat_size-change_2025}.
While SCT ensures termination, it is not clear that recursive programs satisfying SCT are conservative over primitive eliminators.

To solve this open problem, a translation from cyclic to inductive proofs is required that is sensitive to a fine-grained notion of computational behaviour: $\beta$-reduction of (open) terms.
This is due to the fact that, in dependent type theory, since types depend on terms, the definitional equalities that a function/proof satisfies influences which statements type-check.
So, it is only sound to replace a cyclic term with an inductive implementation if we can fully replicate its definitional behaviour.
The unfolding technique seems suitable for this, as it offers a construction of inductive proofs that directly follows the structure of the initial cyclic proof.
This work therefore extends the unfolding technique to a general setting targeting multi-sorted predicate logic.

\paragraph{Contribution.} 
Given an abstract cyclic proof system $\cal{C}$, whose soundness condition corresponds precisely to the size-change termination principle, we do the following: \begin{enumerate}
    \item We identify the notion of a \emph{well-organized} cyclic representation of an infinite proof.
    \item We show that every $\cal{C}$-proof has a well-organized cyclic representation.
    \item We provide a constructive translation from any well-organized cyclic representation in $\cal{C}$ to a finite proof in an induced inductive proof system $\cal{C}_{\sf{ind}}$.
\end{enumerate}
Here \(\cal C_\ind\) is a multi-sorted predicate logic with a well-founded induction rule for each inductive sort of $\cal{C}$.
We illustrate how \(\cal C_\ind\) can be used as an intermediate step by translating a version of cyclic \(k\)th-order Heyting arithmetic (\(\sf{CHA}k\)) via \((\sf{CHA}k)_\ind\) to ordinary \(k\)th-order Heyting arithmetic (\(\sf{HA}k\)) \cite{troelstra_constructivism_1988,van2008realizability,otten_conservativity_2024}.

To obtain these results, we build on the translation from $\sf{CHA}$ to $\sf{HA}$ in the PhD thesis of Wehr \cite{wehr2025cyclic}, developed in collaboration with Afshari and Leigh \cite{afshari_abstract_2022,leigh_gtc_2024,leigh2025unravelling}.
To find a suitable cyclic representation, sequents are \emph{annotated} to obtain \emph{reset proofs} \cite{jungteerapanich2010tableau,Stirling}, which have a local rather than a global soundness condition.
These reset proofs can be used to extract an induction order and suitable induction hypotheses.
In the style of Sprenger \& Dam \cite{SprengerDamComp}, the reset proof is then unfolded until the tree-structure of the back-edges respects the induction order. 

For the final step in their translation, Leigh \& Wehr \cite{leigh2025unravelling} rely on the linearity of the order on natural numbers, which fails for more general inductive sorts.
A key novelty in our work lies here: we circumvent the need for linearity by imposing stronger criteria for progress on reset proofs and refining the induction hypotheses.

A second novelty is our identification of a well-organized cyclic representation (contribution 1), which formally captures cyclic representations that are suitable for direct translation to an inductive proof.
This notion not only improves our understanding of the unfolding technique in general, but also allows one to adapt our translation (contribution 3) to other inductive systems, while treating the existence of a well-organized cyclic representation (contribution 2) as a `black box'; this is particularly desirable since this black box relies on a bespoke and highly technical definition of reset proofs, while the translation is constructive and relatively straightforward. 

\paragraph{Structure of the article.} Section~\ref{sec:cyclic system} defines abstract cyclic proof systems, Section~\ref{sec:multi sorts} introduces multi-sorted predicate logic, and Section~\ref{sec:inductive system} defines the induced inductive system $\cal{C}_{\sf{ind}}$ of a cyclic system $\cal{C}$. 
Section~\ref{sec:examples} sketches the main idea of our translation through examples, while Section~\ref{sec:transformation} provides the definition of a well-organized cyclic representation and its translation to an inductive proof in $\cal{C}_{\sf{ind}}$.
Lastly, Section~\ref{sec:safra} is the most technical part, where we show that every proof in $\cal{C}$ has a well-organized cyclic representation.

%%novalidate

\section{Abstract cyclic proof systems}\label{sec:cyclic system}

We consider an abstract notion of a cyclic proof system that captures both size-change terminating functions and trace-based cyclic proofs. The soundness condition will be described in terms of size-change graphs \cite{lee_size-change_2001,ikebuchi_cyclic_2025}:
\begin{definition}[size-change graph]
    For finite sets \(X,Y\), a \emph{size-change graph} ${G:X\to Y}$ is a bipartite graph from $X$ to $Y$ where edges are labeled as either $\geq$ \emph{(preserving)} or $>$ \emph{(progressing)}.
    A \emph{trace} through an infinite sequence of size-change graphs \((G_i:{X_i\to X_{i+1}})_{i<\omega}\) consists of an index \(k\in\bb N\) and for every \(i\ge k\) a node \(x_i\in X_i\) such that $x_i$ and \(x_{i+1}\) are connected by an edge in $G_i$; the sequence is \emph{progressing} if it has a trace in which the connecting edge is progressing infinitely often.
\end{definition}

\begin{example}\label{exam:run}
    We consider some recursive functions and their corresponding size-change graphs.
    The first function $\sf{Ack}$ is the Ackermann function; the second function $\sf{inter}$ interleaves two lists, sending for example \((n_0^{},n_0',n_0'')\) and \((n_1^{},n_1')\) to \((n_0^{},n_1^{},n_0',n_1',n_0'')\): \begin{align*}
        &\sf{Ack}:\bb N\times\bb N\to\bb N, &
        \hfill&\sf{inter}:\sf{List}(\bb N)\times\sf{List}(\bb N)\to\sf{List}(\bb N), \\
        &\sf{Ack}(0,x_1)\coloneqq\sf{suc}(x_1), &
        &\sf{inter}((),x_1^{})\coloneqq x_1^{}, \\
        &\sf{Ack}(\sf{suc}(x_0'),0)\coloneqq\sf{Ack}(x_0',1),~~~~~~~~~~~~~~~~~~~~ &
        &\sf{inter}((n_0,x_0'),x_1^{})\coloneqq(n_0,\sf{inter}(x_1^{},x_0')). \\
        &\sf{Ack}(\sf{suc}(x_0'),\sf{suc}(x_1'))\coloneqq \mathrlap{\sf{Ack}(x_0',\sf{Ack}(\sf{suc}(x_0'),x_1'));}
    \intertext{Lastly, we consider two functions $f$ and $g$ defined by mutual recursion:}
        &f:\bb N\times\bb N\to\bb N, &
        &g:\bb N\to\bb N, \\
        &f(x_0,x_1)\coloneqq g(\sf{min}(x_0,x_1)), &
        &g(0)\coloneqq1, \\
        & &
        &g(\sf{suc}(x_0'))\coloneqq f(x_0',100)+1.
    \end{align*}
    The sets \(\bb N\) and \(\sf{List}(\bb N)\) both have a transitive well-founded relation: for \(\bb N\) we have the usual order defined by \((n<m)\defeq\exists_{x:\bb N}\,(\suc(n+x)=m)\), and for \(\sf{List}(\bb N)\) we have the strict suffix order generated by \(x<(n,x)\).
    For each recursive call in each function, this well-founded relation induces a size-change graph between the inputs \(\overline x\) of the initial function and the inputs \(\overline y\) of the recursive call, as shown below: \begin{gather*}
    \begin{tikzpicture}[yscale=0.8,remember picture]
        \node at (0,0) {\(\mathsf{Ack}(\subnode{x0a}{x_0},\subnode{x1a}{x_1})\)};
        \node at (0,1) {\(\mathsf{Ack}(\subnode{y0a}{y_0},\subnode{y1a}{y_1})\)};
        \draw[->,green] (x0a.north) -- node[left] {\(\rotatebox{90}{\(>\)}\)} (y0a.south);
    \end{tikzpicture}~~~
    \begin{tikzpicture}[yscale=0.8,remember picture]
        \node at (0,0) {\(\mathsf{Ack}(\subnode{x0a'}{x_0},\subnode{x1a'}{x_1})\)};
        \node at (0,1) {\(\mathsf{Ack}(\subnode{y0a'}{y_0},\subnode{y1a'}{y_1})\)};
        \draw[->,green] (x0a'.north) -- node[left] {\(\rotatebox{90}{\(>\)}\)} (y0a'.south);
    \end{tikzpicture}~~~
    \begin{tikzpicture}[yscale=0.8,remember picture]
        \node at (0,0) {\(\mathsf{Ack}(\subnode{x0b}{x_0},\subnode{x1b}{x_1})\)};
        \node at (0,1) {\(\mathsf{Ack}(\subnode{y0b}{y_0},\subnode{y1b}{y_1})\)};
        \draw[->]       (x0b.north) -- node[left]  {\(\rotatebox{90}{\(\ge\)}\)} (y0b.south);
        \draw[->,green] (x1b.north) -- node[right] {\(\rotatebox{90}{\(>\)}\)}   (y1b.south);
    \end{tikzpicture}~~~
    \begin{tikzpicture}[yscale=0.8,remember picture]
        \node at (0,0) {\(\mathsf{inter}(\subnode{x0c}{x_0},\subnode{x1c}{x_1})\)};
        \node at (0,1) {\(\mathsf{inter}(\subnode{y0c}{y_0},\subnode{y1c}{y_1})\)};
        \draw[->,green] (x0c.north) -- node[pos=.3,left=7pt]  {\(\rotatebox{90}{\(>\)}\)}   (y1c.south);
        \draw[->]       (x1c.north) -- node[pos=.3,right=7pt] {\(\rotatebox{90}{\(\ge\)}\)} (y0c.south);
    \end{tikzpicture}~~~
    \begin{tikzpicture}[yscale=0.8,remember picture]
        \node at (0,0) {\(f(\subnode{x0d}{x_0},\subnode{x1d}{x_1})\)};
        \node at (0,1) {\(g(\subnode{y0d}{y_0})\)};
        \draw[->,in=90] (x0d.north) to[out=60,in=-90] node[left=7pt] {\(\rotatebox{90}{\(\ge\)}\)} (y0d.south);
        \draw[->,in=-90] (x1d.north) to[out=120,in=-90] node[right=7pt] {\(\rotatebox{90}{\(\ge\)}\)} (y0d.south);
    \end{tikzpicture}~~~
    \begin{tikzpicture}[yscale=0.8,remember picture]
        \node at (0,0) {\(g(\subnode{x0e}{x_0})\)};
        \node at (0,1) {\(f(\subnode{y0e}{y_0},\subnode{y1e}{y_1})\)};
        \draw[->,green] (x0e.north) -- node[left] {\(\rotatebox{90}{\(>\)}\)} (y0e.south);
    \end{tikzpicture}
\end{gather*}
    For example, in the call \(f(x_0,x_1)\coloneqq g(\sf{min}(x_0,x_1))\), we have two initial inputs \(x_0\) and \(x_1\), and one recursive input \(y_0\coloneqq\sf{min}(x_0,x_1)\), which induces two preserving edges \(x_0\ge y_0\) and \(x_1\ge y_0\). In \(g(\sf{suc}(x'_0))\coloneqq f(x_0',100)\), we have initial input \(x_0\coloneqq\suc(x_0')\) and recursive inputs \(y_0\coloneqq x_0'\) and \(y_1\coloneqq100\), which induces the progressing edge \(x_0>y_0\). 
    
    For each of these recursive functions, one easily checks it is \emph{size-change terminating}, meaning that every infinite sequence of recursive calls induces a progressing sequence of graphs. As shown in \cite{lee_size-change_2001}, this means that these functions are indeed terminating. 
\end{example}

\begin{definition}[cyclic proof system]
    A \emph{proof system} consists of a set $\sf{Judg}$ \emph{(judgments)}, a set $\sf{Rule}$ \emph{(rule instances)}, and for each $r\in\sf{Rule}$ a finite list of \emph{premises} $\sf{Prem}(r)\in\sf{List}(\sf{Judg})$ and a \emph{conclusion} $\sf{concl}(r)\in\sf{Judg}$.
    A \emph{cyclic proof system} has additional structure: we have a set \(\sf{Sort}\) \emph{(trace sorts)}, each judgment $A\in\sf{Judg}$ has a list of \emph{trace objects} $\sf{Ob}(A)\in\sf{List}(\sf{Sort})$, and for each rule $r\in\sf{Rule}$ and each $i<|\sf{Prem}(r)|$ we have a \emph{size-change graph} $\sf{Graph}(r)_i:|\sf{Ob}(\sf{concl}(r))|\to|\sf{Ob}(\sf{Prem}(r)_i)|$ that only has edges between indices that correspond to the same trace sort.
    Here \(|{-}|\) gives the length of the list, which we view as a set \(n=\{0,\dots,n-1\}\).
\end{definition}
Many cyclic systems only have one trace sort. For example, in arithmetic (see e.g. \autoref{exam:cyc arith}) this is the sort of natural numbers, and for \(\mu\)-calculi this is the sort of formulas. For recursive functions, we have a trace sort for each inductive argument.
\begin{example}\label{exam:run1}
    The functions in \autoref{exam:run} form a cyclic proof system with two trace sorts \(\bb N\) and \(\sf{List}(\bb N)\) and four judgments \(\sf{Ack}{\downarrow},\sf{inter}{\downarrow},f{\downarrow}\) and \(g{\downarrow}\), whose trace objects are given by the domain of the corresponding function (for example, \(\sf{Ob}(\sf{Ack}{\downarrow})=(\bb N,\bb N)\)). Moreover, for each function, we have a rule: \begin{gather*}
        \prftree[r]{\(\sf{Ack}\),}
            {\sf{Ack}{\downarrow}}
            {\sf{Ack}{\downarrow}}
            {\sf{Ack}{\downarrow}}
            {\sf{Ack}{\downarrow}} \qquad
        \prftree[r]{\(\sf{inter}\),}
            {\sf{inter}{\downarrow}}
            {\sf{inter}{\downarrow}} \qquad
        \prftree[r]{\(f\),}
            {g{\downarrow}}
            {f{\downarrow}} \qquad
        \prftree[r]{\(g\),}
            {f{\downarrow}}
            {g{\downarrow}}
    \end{gather*} where the size-change graphs of the premises are precisely those drawn in \autoref{exam:run}. Intuitively, one can view these rules as stating that a function terminates whenever each function occurring in its recursive calls terminates. 
\end{example}

Derivations in a (cyclic) proof system will be trees labeled by rules.
For us, a \emph{tree} consists of a set \(\sf{Node}\) (the \emph{nodes}), together with a \emph{root} \(\sf{root}\in\sf{Node}\), and for every \(n\in\sf{Node}\) a non-repeating list of \emph{children} \(\sf{Child}(n)\in\sf{List}(\sf{Node})\), such that for any node \(n\) there exists a unique path from the root to \(n\).
A \emph{path} is a sequence of nodes \((n_i)_{i<m}\) where each \(n_{i+1}\) is a child of \(n_i\), and a \emph{branch} is a maximal path.
We say that \(n\) is \emph{below} \(n'\) if \(n\) is on the path from the root to \(n'\). We then define a \emph{derivation} of a proof system as a (possibly infinite) tree where every node \(n\) is labeled by a rule $\sf{rule}(n)\in\sf{Rule}$ with \(|\sf{Prem}(\sf{rule}(n))|=|\sf{Child}(n)|\) and \(\sf{Prem}(\sf{rule}(n))_i=\sf{concl}(\sf{rule}(\sf{Child}(n)_i))\).
    For brevity, we define \(\sf{judg}(n)\coloneqq\sf{concl}(\sf{rule}(n))\).
    
  A derivation is \emph{regular} if it has only finitely many distinct subderivations up to isomorphism. (A \emph{subderivation} is generated by the upset of some node, and an \emph{isomorphism} is a bijection that respects the root node, the children, and the rules.) Note that, by definition, each infinite branch in a derivation induces a unique sequence of size-change graphs. We can therefore impose the following soundness condition:

\begin{definition}[proof] 
    A regular derivation is a \emph{proof} if the induced sequence of size-change graphs of every infinite branch is progressing.
\end{definition}
\begin{example}\label{exam:run2}
    In the system of \autoref{exam:run1}, every judgment $h{\downarrow}$ has a unique derivation, and this derivation is a proof precisely because the corresponding function $h$ from \autoref{exam:run} is size-change terminating.
\end{example}
\begin{remark}\label{remark:parity}
    This abstract notion of cyclic proofs is more general than it might appear. It covers cyclic systems with a `B\"uchi' soundness condition (such as Sprenger \& Dam's system with ordinal approximations \cite{SprengerDamComp} or the cyclic system for arithmetic in \autoref{exam:cyc arith}), but also systems with a `parity' soundness condition, such as those encountered in the modal \(\mu\)-calculus \cite{NW96}. To formulate $\mu$- and \(\nu\)-traces in terms of size-change graphs, one can take tuples of formulas and bound variables as trace objects \cite[Definition 3.4]{afshari_abstract_2022}.
\end{remark}

Crucially, regular derivations can be represented by a finite subtree with `back-edges', which connect a leaf of the subtree (a `bud') to a lower node (its `sprout'):
\begin{definition}[cyclic representation]
    A \emph{cyclic representation} of a regular derivation consists of a downward closed set of nodes \(\sf{RepNode}\subseteq\sf{Node}\).
    In addition, there is a subset \(\sf{Bud}\subseteq\sf{RepNode}\) (the \emph{buds}) such that for every \(n\in\sf{RepNode}\): either \(n\in\sf{Bud}\) and none of its children are in \(\sf{RepNode}\), or \(n\notin\sf{Bud}\) and all of its children are in \(\sf{RepNode}\).
    Every bud \(b\in\sf{Bud}\) has an associated \emph{sprout/companion} \(\sf{sprout}(b)\in\sf{RepNode}\) strictly below \(b\) such that \(b\) and \(\sf{sprout}(b)\) have isomorphic subderivations.
\end{definition}
%%novalidate

\section{Multi-sorted predicate logic}\label{sec:multi sorts}

In this section we introduce multi-sorted predicate logic. We will use this to define the induced inductive proof system \(\cal C_\ind\) for any cyclic proof system \(\cal C\), but also to give a more complicated example of a cyclic proof system (see \autoref{exam:cyc arith}).

To define multi-sorted predicate logic, we start with a \emph{language}, which consists of a set of \emph{sorts}, a set of \emph{relation symbols}, which come with a finite list of input sorts, and a set of \emph{function symbols}, which come with a finite list of input sorts and an output sort.

For example, given a fixed \(1\le k\le\omega\) we consider the \emph{language of \(k\)th order arithmetic} that has \(k\) sorts, which we suggestively call \(\{\cal P^i(\bb N)\}_{i<k}\) and should be viewed as iterated powersets of the natural numbers. For every sort of the form \(\cal P^{i+1}(\bb N)\) we have a binary relation symbol \(\in_{\cal P^{i+1}(\bb N)}\) with input sorts \(\cal P^i(\bb N)\) and \(\cal P^{i+1}(\bb N)\); in addition, we have a nullary function symbol \(0\), a unary function symbol \(\suc\), and binary function symbols \(+\) and \(\times\), each with input and output sort \(\bb N\) (where \(\bb N\defeq \cal P^{0}(\bb N)\)).

Terms and formulas for such a language can only be defined after we specify the sorts of free variables. We fix an infinite list of variables, and define a \emph{variable context} as a finite set \(\cal X\) consisting of pairs of a variable and a sort such that the variables are distinct; we will denote such a variable context as \(x_0:S_0,\dots,x_{n-1}:S_{n-1}\) or \(\overline{x:S}\).
Given a variable context \(\cal X\), we define the terms of a sort $S$ by mutual induction: \begin{align*}
    a&\Coloneqq x\,|\,f(\overline{a}),\tag{terms of sort \(S\) in \(\cal X\)}
\end{align*} where \((x:S)\in\cal X\), \(f\) is a function symbol with output sort \(S\), and \(\overline{a}\) is a list of terms that have the input sorts of \(f\).
Moreover, we define the formulas with free variables in \(\cal X\): \begin{align*}
    \phi&\Coloneqq R(\overline a)\,|\,\bot\,|\,\phi\lor\phi'\,|\,\phi\land\phi'\,|\,\phi\to\phi'\,|\,\exists_{x:S}\,\psi\,|\,\forall_{x:S}\,\psi\,|\,(a=_Sa'), \tag{formulas in \(\cal X\)}
\end{align*} where \(R\) is a relation symbol, the \(\overline a\) are terms of the input sorts of \(R\) in \(\cal X\), \(S\) is a sort, \(\psi\) is a formula in the extended variable context \(\cal X,x:S\), and \(a\) and \(a'\) are terms of \(S\) in \(\cal X\).

A judgment of multi-sorted predicate logic is a sequent of the form \(\cal X;\Gamma\vdash\delta\), where \(\cal X\) is a variable context, \(\Gamma\) is a finite set of formulas in \(\cal X\) (the \emph{formula context}), and \(\delta\) is a formula in \(\cal X\).
We have the following standard natural deduction rules: \begin{rules}{intuitionistic multi-sorted predicate logic}
    \prfstack[r]{id}
        {\cal X;\Gamma,\phi\vdash\phi} ~\;
    \prfstack[r]{cut} 
        {\cal X;\Gamma,\phi\vdash\delta}
        {\cal X;\Gamma\vdash\phi}
        {\cal X;\Gamma\vdash\delta} ~\;
    \prfstack[r]{weaken}
        {\cal X;\Gamma\vdash\delta}
        {\cal X,\cal X';\Gamma,\Gamma'\vdash\delta} ~\;
    \prfstack[r]{\(\bot\)-elim}
        {\cal X;\Gamma\vdash\bot}
        {\cal X;\Gamma\vdash\delta}\\[1ex]
    \prfstack[r]{\(\lor\)-intro\(_i\)}
        {\cal X;\Gamma\vdash\phi_i}
        {\cal X;\Gamma\vdash\phi_0\lor\phi_1} ~~
    \prfstack[r]{\(\land\)-intro}
        {\cal X;\Gamma\vdash\phi_0}
        {\cal X;\Gamma\vdash\phi_1}
        {\cal X;\Gamma\vdash\phi_0\land\phi_1} ~~
    \prfstack[r]{\(\to\)-intro}
        {\cal X;\Gamma,\phi_0\vdash\phi_1}
        {\cal X;\Gamma\vdash\phi_0\to\phi_1} \\[1ex]
    \prfstack[r]{\(\lor\)-\makesize[l]{elim}{intro\(_i\)}}
            {\cal X;\Gamma\vdash\phi_0\lor\phi_1}
            {\cal X;\Gamma,\phi_0\vdash\delta}
            {\cal X;\Gamma,\phi_1\vdash\delta}
            {\cal X;\Gamma\vdash\delta} ~~
    \prfstack[r]{\(\land\)-\makesize[l]{elim\(_i\)}{intro,}}
        {\cal X;\Gamma\vdash\phi_0\land\phi_1}
        {\cal X;\Gamma\vdash\phi_i} ~~
    \prfstack[r]{\(\to\)-\makesize[l]{elim}{intro}}
        {\cal X;\Gamma\vdash\phi_0\to\phi_1}
        {\cal X;\Gamma\vdash\phi_0}
        {\cal X;\Gamma\vdash\phi_1} \\[1ex]
    \prfstack[r]{\(\exists\)-intro}
        {\cal X;\Gamma\vdash\psi[a/x]}
        {\cal X;\Gamma\vdash\exists_{x:S}\,\psi}~~~
    \prfstack[r]{\(\forall\)-intro}
        {\cal X,x:S;\Gamma\vdash\psi}
        {\cal X;\Gamma\vdash\forall_{x:S}\,\psi} ~~~
    \prfstack[r]{\(=\)-intro}
        {\cal X,x:S;\Gamma\vdash x=_Sx} \\[1ex]
    \prfstack[r]{\(\exists\)-elim}
        {\cal X;\Gamma\vdash\exists_{x:S}\,\psi}
        {\cal X,x:S;\Gamma,\psi\vdash\delta}
        {\cal X;\Gamma\vdash\delta} \quad
    \prfstack[r]{\(\forall\)-elim}
        {\cal X;\Gamma\vdash\forall_{x:S}\,\psi}
        {\cal X;\Gamma\vdash\psi[a/x]} \quad
    \prfstack[r]{\(=\)-elim}
        {\cal X;\Gamma\vdash a=_Sb}
        {\cal X;\Gamma\vdash\delta[a]}
        {\cal X;\Gamma\vdash\delta[b]}
\end{rules}
% Here cut is only added for convenience; it is already derivable from \(\to\)-intro/elim.
% Similarly, weaken is already admissible using these rules, but this need not be the case for extensions of the system.
%\emph{Classical multi-sorted predicate logic} can be obtained by adding a rule for double negation elimination.

\begin{example}[arithmetic]\label{exam:arith}
    We define \(\sf{HA}k\) (\(k\)th-order Heyting arithmetic \cite{troelstra_constructivism_1988,van2008realizability,otten_conservativity_2024}) by taking the language of \(k\)th-order arithmetic defined above, and extending intuitionistic multi-sorted predicate logic with standard axioms for natural numbers and powersets (see Appendix \ref{app: rules of CHAk}).
    For any context \(\cal X\) and formula \(\psi\) in \(\cal X,x:\bb N\), we also add the axiom \[
        \psi[0/x]\to\forall_{x:\bb N}\,(\psi\to\psi[\sf{suc}(x)/x])\to\forall_{x:\bb N}\,\psi. \tag{induction}
    \]
   
\end{example}
\begin{example}[cyclic arithmetic]\label{exam:cyc arith}
We can obtain a cyclic proof system \(\sf{CHA}k\) from \(\sf{HA}k\) by replacing the induction scheme with the following case distinction rule: \[
        \prftree[r]{case\(_x\).}
            {\cal X;\Gamma[0/x]\vdash\delta[0/x]}
            {\cal X,x:\bb N;\Gamma[\sf{suc}(x)/x]\vdash\delta[\sf{suc}(x)/x]}
            {\cal X,x:\bb N;\Gamma\vdash\delta}
    \]
    In addition, we bestow \(\sf{CHA}k\) with the following trace information: \begin{itemize}
        \item the only \emph{trace sort} is \(\bb N\);
        \item the \emph{trace objects} of a sequent \(\cal X;\Gamma\vdash\delta\) are the variables in \(\cal X\) of sort \(\bb N\) (in the order given by the ambient infinite list of variables);
        \item the \emph{size-change graph} from a conclusion to a premise connects indices iff they correspond to the same variable; an edge from/to \(x\) progresses iff the rule is case\(_x\).
    \end{itemize}
    The following (cyclic representation of a) \(\sf{CHA}k\)-proof derives the induction axiom: \[\prfinterspace = 0.5em
      \prftree[r,double]{\(\forall/{\to}\)-intro.}
        {\prftree[r]{case\(_x\)}
          {\prftree[r]{id}
            {\cdots\vdash\psi[0/x]}}
          {\prftree[r,double]{\(\forall/{\to}\)-elim}
            {\prftree[r]{id}
              {\cdots\vdash\forall_{x:\mathbb{N}}\,(\psi\to\psi[\mathsf{suc}(x)/x])}}
            {\prfassumption{\tikzmarknode{bud}{\cdots\vdash\psi}}}
            {\cdots\vdash\psi[\mathsf{suc}(x)/x]}}
          {\tikzmarknode{comp}{\mathcal{X},x:\mathbb{N};\psi[0/x],\forall_{x:\mathbb{N}}\,(\psi\to\psi[\mathsf{suc}(x)/x])\vdash\psi}}}
        {\mathcal{X}\vdash\psi[0/x]\to\forall_{x:\mathbb{N}}\,(\psi\to\psi[\mathsf{suc}(x)/x])\to\forall_{x:\mathbb{N}}\,\psi}
      \begin{tikzpicture}[remember picture, overlay]
        \draw[->, blue, dashed] (bud.east) to[out=0, in=0, looseness=4.1] (comp.east);
    \end{tikzpicture}\] We can thus transform every \(\sf{HA}k\)-proof into a \(\sf{CHA}k\)-proof.
    The converse will follow from our general translation (\autoref{the:main}), as explained in \autoref{exam: flattening CHAind}.
    The system \(\sf{CHA}k\) is inspired by Leigh \& Wehr's system \(\sf{CHA}\) \cite{leigh2025unravelling} (although we use natural deduction instead of sequent calculus), which has a simpler soundness condition than the version of cyclic arithmetic presented by Simpson \cite{simpson2017cyclic}.
\end{example} 

\section{The induced inductive proof system}\label{sec:inductive system}

Fix a cyclic proof system $\cal{C}$.
We will define a new proof system \(\cal C_{\sf{ind}}\) with an explicit induction rule such that every proof in $\cal{C}$ can be translated to a finite proof in \(\cal C_{\sf{ind}}\).
In the previous section, we saw the cyclic proof system \(\sf{CHA}k\), whose judgments are sequents of multi-sorted predicate logic.
For abstract cyclic systems, we do not know anything about the judgments; in order to state the induction rule, we define \(\cal C_{\sf{ind}}\) as an extension of intuitionistic multi-sorted predicate logic in which we can simulate the rules of $\cal{C}$.

The sorts of \(\cal C_\ind\) are the trace sorts of \(\cal C\), and every sort \(S\) has a binary relation \(<_S\) for which we have a transitivity rule and a well-founded induction rule in \(\cal C_{\sf{ind}}\): \begin{rules}{transitive well-founded relation}
    \prfstack[r]{trans}
        {\cal X;\Gamma\vdash a''<_Sa'}
        {\cal X;\Gamma\vdash a'\phantom'<_S a}
        {\cal X;\Gamma\vdash a''<_Sa} ~
    \prfstack[r]{ind\(_x\)}
        {\cal X,x:S;\Gamma,\forall_{x':S,\cal X}\,(x'<_Sx\to(\Gamma\to\delta)[x'/x])\vdash\delta}
        {\cal X,x:S;\Gamma\vdash\delta}
\end{rules}
As usual, we define \((a\le b)\defeq(a<b)\lor(a=b)\).
Moreover, for every judgment \(A\) of \(\cal C\) we have a relation symbol \(A\) in \(\cal C_\ind\) with input sorts \(\sf{Ob}(A)\).
We think of \(A(\overline a)\) as the statement that \(A\) holds for terms \(\overline a\); for example, for the judgment \(\sf{Ack}{\downarrow}\) of \autoref{exam:run1} we view \(\sf{Ack}{\downarrow}(\overline a)\) as stating that $\sf{Ack}$ terminates on inputs \(\overline a\).
Lastly, for every rule \(r\) in \(\cal C\) with conclusion $A$ with trace objects \(\overline S\), premises $A_0',\dots,A_{n-1}'$ with trace objects \(\overline{S_0'},\dots,\overline{S_{n-1}'}\), and size-change graphs \(G_0,\dots,G_{n-1}\), we have the following rule in \(\cal C_{\sf{ind}}\): \begin{rules}{\(\cal C\)-rules}
    \prfstack[r]{$r$}
        {\cal X,\overline{x:S},\overline{y:S_i'};\Gamma,\overline{x}\gtrsim_{G_i}\overline{y}\vdash A_i'(\overline{y}) & \qquad \text{for every $i<n$}}
        {\cal X,\overline{x:S};\Gamma\vdash A(\overline{x})}
\end{rules}
Here \(\overline{x}\gtrsim_{G_i}\overline{y}\) is a list of the inequalities given by the edges in the size-change graph \(G_i\): it contains \(x_{j}>y_{j'}\) iff \((j> j')\in G_i\), and \(x_{j}\ge y_{j'}\) iff \((j\ge j')\in G_i\).

\begin{remark}
    Note that ind$_x$, which does induction on the entire sequent, is interderivable with the following rule, which does induction on only the conclusion: \[
        \prfstack[r]{ind\(^{\sf{conc}}_x\).}
            {\cal X,x:S;\Gamma,\forall_{x':S}\,(x'<_Sx\to\delta[x'/x])\vdash\delta}
            {\cal X;\Gamma\vdash\forall_{x:S}\,\delta}
    \]
    One can derive ind$_x$ by applying ind\(_x^{\sf{conc}}\) to the formula $\delta'\coloneqq\forall_{\cal X}\,(\Gamma\to\delta)$, and ind\(_x^{\sf{conc}}\) by applying ind$_x$ and eliminating the extra \(\forall\) and \(\to\) (see Appendix \ref{app: proof interderivability}).
    The upshot of ind$_x$ over ind\(_x^{\sf{conc}}\) is that it provides a stronger induction hypothesis that can be applied to any variables \(\cal X'\) for which we can show \(\Gamma\vdash\Gamma[\cal X'/\cal X]\).
\end{remark}

\begin{example}[From \((\sf{CHA}k)_\ind\) to \(\sf{HA}k\)]\label{exam: flattening CHAind}
    Note that \(\cal C_{\sf{ind}}\) extends \(\cal C\) with a layer of multi-sorted predicate logic. While this is necessary for abstract cyclic systems, in \((\sf C)\sf{HA}k\) we can already define \(<\) as \((m<n)\defeq\exists_{x:\bb N}\,(\sf{suc}(m+x)=n)\) and prove that it is transitive and well-founded.
    Moreover, since \(\sf{CHA}k\) is already a multi-sorted predicate logic, \((\sf{CHA}k)_\ind\) contains two layers of sequents, where a sequent \(x_0:\bb N,\dots,x_{n-1}:\bb N,\cal X;\Gamma\vdash\delta\) of \(\sf{CHA}k\) can occur within a sequent of \((\sf{CHA}k)_\ind\) as a relation symbol of arity \(n\).
    There is a straightforward translation from judgments of \((\sf C\sf{HA}k)_\ind\) to sequents of \(\sf{HA}k\) that flattens the two layers: it sends an atom \((\overline{x:\bb N},\cal X;\Gamma\vdash\delta)(\overline{a})\) to the formula \(\forall_{\cal X}\,(\Gamma[\overline{a/x}]\to\delta[\overline{a/x}])\).
    In Appendix \ref{app:flattening}, we show that this induces a translation of proofs in \((\sf{CHA}k)_\ind\) to proofs in \(\sf{HA}k\). 
\end{example}
%%novalidate

\section{General idea and case studies}\label{sec:examples}

In this section, we consider some examples of translating a \(\cal C\)-proof \(\pi\) into a finite \(\cal C_\ind\)-proof \(\pi_\ind\).
This allows us to explain the techniques we use in the general translation.

The transformation proceeds in two steps: we first transform $\pi$ into a (possibly infinite) derivation $\pi_\ind^{\sf{inf}}$ in the inductive system $\cal{C}_\sf{ind}$, and then show how the infinite branches of $\smash{\pi_\ind^{\sf{inf}}}$ can be `pruned' by introducing induction hypotheses and using them at a leaf node. The first step is straightforward: we decorate the labels of $\pi$ with variables and inequalities such that every branch $(n_i)_i$ in $\pi$ with corresponding judgments \((A_i)_i\), rules \((r_i)_i\), trace objects \((\overline{S_i})_i\), and size-change graphs $(G_i)_i$, induces a branch in $\pi_\ind^{\sf{inf}}$ of the following form:
\[
    \prftree[r]{$r_0$,}
        {\prftree[r]{$r_1$}
            {\prftree[noline]
                {\vdots}
                {A_2}}
            {A_1}}
        {A_0}
    \qquad\rightsquigarrow\qquad
    \prftree[r]{$r_0$.}
        {\prftree[r]{$r_1$}
            {\prftree[noline]
                {\vdots}
                {\overline{x:S_0},\overline{y:S_1},\overline{z:S_2};\overline{x}\gtrsim_{G_0}\overline{y},\overline{y}\gtrsim_{G_1}\overline{z}\vdash A_2(\overline{z})}}
            {\overline{x:S_0},\overline{y:S_1};\overline{x}\gtrsim_{G_0}\overline{y}\vdash A_1(\overline{y})}}
        {\overline{x:S_0}\vdash A_0(\overline{x})}
\]
For the second step, we use a cyclic representation \(\pi_\sf{cyc}\) of $\pi$ to prune the infinite branches of $\smash{\pi_\ind^{\sf{inf}}}$: we introduce induction hypotheses at sprouts and we apply the corresponding induction hypothesis at buds to `close' the branch.
Choosing a suitable cyclic representation for this is the main difficulty of the translation, and so we use the rest of this section to illustrate the conditions that such a representation should satisfy.

\begin{example}[importance of progress]\label{exam:importprogres}
    Recall the interleave function from \autoref{exam:run}. For the unique proof \(\pi\) of $\sf{inter}{\downarrow}$ (see \autoref{exam:run1}), the first step \(\pi\rightsquigarrow\pi_\ind^{\sf{inf}}\) is given by: 
    \vspace{4pt}
    \[
        \prftree[r]{$\sf{inter}$,}
            {\prftree[r]{$\sf{inter}$}
                {\prftree
                    {\vdots}
                    {\sf{inter}{\downarrow}}}
                {\sf{inter}{\downarrow}}}
            {\sf{inter}{\downarrow}}
        ~~\rightsquigarrow\!\!\!\!\!\!\!\!
        \smash{\prftree[r]{$\sf{inter}$.}
            {\prftree[r]{$\sf{inter}$}
                {\prftree
                    {\vdots}
                    {\overline{x,y,z:\sf{List}(\bb N)};x_{0}>y_{1},x_{1}\ge y_{0},y_{0}>z_{1},y_{1}\ge z_{0}\vdash\sf{inter}(z_{0},z_{1}){\downarrow}}}
                {\overline{x,y:\sf{List}(\bb N)};x_0>y_1,x_1\ge y_0\vdash\sf{inter}(y_0,y_1){\downarrow}}}
            {\overline{x:\sf{List}(\bb N)}\vdash\sf{inter}(x_0,x_1){\downarrow}}}
    \] The infinite proof $\pi$ has various cyclic representations.
    Two of them are depicted below, 
    \begin{minipage}[c]{0.55\textwidth}\vspace{2pt} 
        \setlength{\parindent}{1.5em}\noindent
     where we have visualized the underlying size-change graphs rather than the nodes of $\pi$ and we have used dotted arrows to indicate a back-edge from a bud to a sprout.
    We can use the second representation to turn the infinite derivation \(\pi_\ind^{\sf{inf}}\) into a proof \(\pi_\ind\) by introducing an induction hypothesis \(\phi(x_0)\) at the sprout node:
    \end{minipage}\hfill
    \begin{minipage}[c]{0.45\textwidth}
        \centering \vspace*{-25pt}
  \begin{tikzpicture}[yscale=0.8,remember picture, baseline={([yshift=-.5ex]current bounding box.center)}]
            \node (xj) at (0,0) {\(\mathsf{inter}(\subnode{x0j}{x_0},\subnode{x1j}{x_1}){\downarrow}\)};
            \node (yj) at (0,1) {\(\mathsf{inter}(\subnode{y0j}{y_0},\subnode{y1j}{y_1}){\downarrow}\)};
            \draw[->,green] (x0j.north) to node[pos=.3,left=7pt]  {\(\rotatebox{90}{\(>\)}\)}   (y1j.south);
            \draw[->]       (x1j.north) to node[pos=.3,right=7pt] {\(\rotatebox{90}{\(\ge\)}\)} (y0j.south);
            \path [->,in=180,out=180, red, dashed] (yj) edge (xj);
        \end{tikzpicture}
        \smash[b]{\begin{tikzpicture}[yscale=0.8,remember picture, baseline={([yshift=-.5ex]current bounding box.center)}]
            \node (xk) at (0,0) {\(\mathsf{inter}(\subnode{x0k}{x_0},\subnode{x1k}{x_1}){\downarrow}\)};
            \node (yk) at (0,1) {\(\mathsf{inter}(\subnode{y0k}{y_0},\subnode{y1k}{y_1}){\downarrow}\)};
            \node (zk) at (0,2) {\(\mathsf{inter}(\subnode{z0k}{z_0},\subnode{z1k}{z_1}){\downarrow}\)};
            \draw[->,green] (x0k.north) to node[pos=.3,left=7pt]  {\(\rotatebox{90}{\(>\)}\)}   (y1k.south);
            \draw[->]       (x1k.north) to node[pos=.3,right=7pt] {\(\rotatebox{90}{\(\ge\)}\)} (y0k.south);
            \draw[->,green] (y0k.north) to node[pos=.3,left=7pt]  {\(\rotatebox{90}{\(>\)}\)}   (z1k.south);
            \draw[->]       (y1k.north) to node[pos=.3,right=7pt] {\(\rotatebox{90}{\(\ge\)}\)} (z0k.south);
            \path [->,in=10,out=-10, blue, dashed] (zk) edge (xk);
        \end{tikzpicture}}
    \end{minipage}
    
    \[
        \prftree[r]{ind\(_{x_0}\),}
            {\prftree[r]{$\sf{inter}$}
                {\prftree[r]{$\sf{inter}$}
                    {\prftree
                        {\text{(use \(\phi(x_0)\) on \(x_{0}'\coloneqq z_{0}\) and \(x_{1}'\coloneqq z_{1}\))}}
                        {\overline{x,y,z:\sf{List}(\bb N)};\phi(x_0),x_{0}>y_{1},x_{1}\ge y_{0},y_{0}>z_{1},y_{1}\ge z_{0}\vdash\sf{inter}(z_0,z_1){\downarrow}}}
                    {\overline{x,y:\sf{List}(\bb N)};\phi(x_0),x_0>y_1,x_1\ge y_0\vdash\sf{inter}(y_0,y_1){\downarrow}}}
                {\overline{x:\sf{List}(\bb N)};\phi(x_0)\vdash\sf{inter}(x_0,x_1){\downarrow}}}
            {\overline{x:\sf{List}(\bb N)}\vdash\sf{inter}(x_0,x_1){\downarrow}}
    \]
    where \(\phi(x_0)\coloneqq\forall_{x_0',x_1':\sf{List}(\bb N)}\,(x_0'<x_0\to\sf{inter}(x_0',x_1'){\downarrow})\). Note that $\phi(x_0)$ can be applied because the first variable at the bud (\(z_0\)) is strictly smaller than the first variable at the sprout (\(x_0\)).
    In fact, the same is true for the second variable, so we could have introduced an induction hypothesis for \(x_1\) as well.
     The first cyclic representation is not suitable, since it does not have an input that `progresses' in this way before the back-edge.
\end{example}

\begin{example}[importance of induction order]\label{exam:importorder}
    Consider the Ackermann function from \autoref{exam:run}. The unique proof $\pi$ of $\sf{Ack}{\downarrow}$ has the following cyclic representation: \[
        \begin{tikzpicture}[remember picture, baseline={([yshift=-1.5ex]current bounding box.center)}, xscale=1.5]
            \node (xl) at  (0,0) {\(\mathsf{Ack}(\subnode{x0l}{x_0},\subnode{x1l}{x_1}){\downarrow}\)};
            \node (yl) at (-2.5,1) {\(\mathsf{Ack}(\subnode{y0l}{y_0},\subnode{y1l}{y_1}){\downarrow}\)};
            \node (yl') at (0,1) {\(\mathsf{Ack}(\subnode{y0l'}{y_0},\subnode{y1l'}{y_1}){\downarrow}\)};
            \node (yl'') at (2.5,1) {\(\mathsf{Ack}(\subnode{y0l''}{y_0},\subnode{y1l''}{y_1}){\downarrow}\)};
            \draw[->,green,out=180,in=-30] (x0l.north) to node[left=8pt] {\(\rotatebox{160}{\(>\)}\)} (y0l.south);
            \draw[->,green] (x0l.north) to node[left] {\(\rotatebox{90}{\(>\)}\)} (y0l'.south);
            \draw[->,in =-150]       (x0l.north) to node[above left=-1pt] {\(\ge\)} (y0l''.south);
            \draw[->,green, out=0,in =-150] (x1l.north) to node[right=8pt] {\(\rotatebox{20}{\(>\)}\)} (y1l''.south);
            \path [->,in=188,out=180, blue, dashed] (yl) edge node[below left] {\(x_0\)} (xl);
            \path [->,in=188,out=180, blue, dashed] (yl') edge node[above = 11pt] {\(\smash{x_0}\)} (xl);
            \path [->,in=-8,out=0, red, dashed] (yl'') edge node[below right] {\(x_1\)} (xl);
        \end{tikzpicture}\hspace*{-1cm}
    \]
    Here we annotated each back-edge by an input on which it makes progress. At the sprout, we can introduce induction hypotheses for both \(x_0\) and \(x_1\), but the order in which we do matters.
    If we introduce the induction hypothesis for \(x_0\) first, we can transform \(\pi_\ind^{\sf{inf}}\) to: \begin{gather*}
        \prftree[r]{ind\(_{x_0}\),}
            {\prftree[r]{ind\(_{x_1}\)}
                {\prftree[r]{\(\sf{Ack}\)}
                    {\prftree
                        {\text{(use \(\phi_0(x_0)\) on \(x_i'\coloneqq y_i\))}}
                        {\dots,x_0>y_0\vdash \sf{Ack}(y_0,y_1){\downarrow}}}
                    {\prftree
                        {\text{(use \(\phi_1(x_1)\) on \(x_i'\coloneqq y_i\))}}
                        {\dots,x_0\ge y_0,x_1>y_1\vdash \sf{Ack}(y_0,y_1){\downarrow}}}
                    {\overline{x:\bb N};\phi_0(x_0),\phi_1(x_1)\vdash \sf{Ack}(x_0,x_1){\downarrow}}}
                {\overline{x:\bb N};\phi_0(x_0)\vdash \sf{Ack}(x_0,x_1){\downarrow}}}
            {\overline{x:\bb N}\vdash \sf{Ack}(x_0,x_1){\downarrow}}
\end{gather*}
where the induction hypotheses are defined by
        \begin{align*}
            \phi_0(x_0)&\coloneqq\forall_{x_0',x_1':\bb N}\,(x_0'<x_0\to\sf{Ack}(x_0',x_1'){\downarrow}), \\
            \phi_1(x_1)&\coloneqq\forall_{x_0',x_1':\bb N}\,(x_1'<x_1\to\phi_0(x_0')\to\sf{Ack}(x_0',x_1'){\downarrow}).
        \end{align*}
    Note that \(\phi_1(x_1)\) has an additional restriction: we can only use it on \(x_0',x_1'\) to deduce \(\sf{Ack}(x_0',x_1'){\downarrow}\) if we can show \(\phi_0(x_0')\).
    This is possible if \(x_0'\le x_0\) since \(\phi_0(x_0)\) only has one free variable \((x_0)\) that functions as an upper bound for the universal quantifier.
    Also note that the other order of adding induction hypotheses would not work: the back-edge for \(x_1\) `preserves' \(x_0\), but the back-edge for \(x_0\) does not preserve \(x_1\) in this way.
\end{example}

\begin{example}[importance of forgetting and remembering]\label{exam:importforget}
    The previous example shows that the more variables, inequalities, and induction hypotheses are already assumed, the weaker newly introduced induction hypotheses become (note the dependency of $\phi_1(x_1)$ on $\phi_0(x_0)$). 
    This is because an induction hypothesis states that the current sequent holds 
    \par\noindent
    \begin{minipage}[c]{0.8\textwidth}\vspace{2pt} 
        \setlength{\parindent}{1.5em}\noindent
        for smaller values, and the validity of \(\cal X,\cal X';\Gamma,\Gamma'\vdash\delta\) is generally weaker than the validity of \(\cal X;\Gamma\vdash\delta\).
    It is therefore important that we can forget (that is, weaken away) redundant information. 
    
    At the same time, certain information must be kept. For example, consider the two forms of cyclic representations on the right. In the first, we may weaken away the induction hypothesis for \(x_0\) before introducing the induction hypothesis for \(x_1\), because we will not need it above the sprout of \(x_1\). In the second example, however, the induction hypothesis for \(x_0\) must be kept and so the back-edge for \(x_1\) has to preserve \(x_0\).
    \end{minipage}\hfill
    \begin{minipage}[c]{0.2\textwidth}
        \centering
        \begin{tikzpicture}[xscale = 0.22, yscale=0.25]
            \node (n0)  at (1, 0) {\(\bullet\)};
            \node (n1)  at (2, 1) {\(\bullet\)};
            \node (n2)  at (3, 2) {\(\bullet\)};
            \node (n3)  at (2, 3) {\(\bullet\)};
            \node (n3') at (4, 3) {\(\bullet\)};
            \node (n4)  at (5, 4) {\(\bullet\)};
            \path (n0)  edge (n1);
            \path (n1)  edge (n2);
            \path (n2)  edge (n3);
            \path (n2)  edge (n3');
            \path (n3') edge (n4);
            \path [->,out=180,in=180, blue, dashed] (n3) edge node[left] {\(x_0\)} (n0);
            \path [->,out=0,in=0, red, dashed] (n4) edge node[below right] {\(x_1\)} (n3');
        \end{tikzpicture} \\[-5pt]
        \begin{tikzpicture}[xscale =0.22, yscale = 0.25]
            \node (n0)  at (1, 0) {\(\bullet\)};
            \node (n1)  at (2, 1) {\(\bullet\)};
            \node (n2)  at (3, 2) {\(\bullet\)};
            \node (n3)  at (2, 3) {\(\bullet\)};
            \node (n3') at (4, 3) {\(\bullet\)};
            \node (n4)  at (5, 4) {\(\bullet\)};
            \path (n0)  edge (n1);
            \path (n1)  edge (n2);
            \path (n2)  edge (n3);
            \path (n2)  edge (n3');
            \path (n3') edge (n4);
            \path [->,out=180,in=180, blue, dashed] (n3) edge node[left] {\(x_0\)} (n0);
            \path [->,out=0,in=0, red, dashed] (n4) edge node[below right] {\(x_1\)} (n1);
        \end{tikzpicture}
    \end{minipage}
    \par\noindent
    \begin{minipage}[c]{0.6\textwidth}\vspace{2pt} 
        \setlength{\parindent}{1.5em}\indent
    Besides remembering induction hypotheses, it can be crucial to remember variables and inequalities.
    Consider the cyclic representation depicted on the right (for some system with judgments \(A\) and \(B\)).
    If \(x_0,y_0,y_1\), and the inequalities \(y_0< x_0\) and \(y_1<x_0\) are present in the sequent at the sprout, then we can introduce the following induction hypothesis: 
    \end{minipage}\hfill
    \begin{minipage}[c]{0.4\textwidth}
        \centering
       \begin{tikzpicture}[remember picture,yscale=0.75,xscale=1.125]
        \node (xm) at (0,0) {\(B(\subnode{x0m}{x_0})\)};
        \node (ym) at (0,1) {\(A(\subnode{y0m}{y_0},\subnode{y1m}{y_1})\)};
        \node (zm) at (-1,2) {\(A(\subnode{z0m}{z_0},\subnode{z1m}{z_1})\)};
        \node (zm') at (1,2) {\(A(\subnode{z0m'}{z_0},\subnode{z1m'}{z_1})\)};
        \draw[->,green] (x0m.north) to node[left=5pt]  {\(\rotatebox{90}{\(>\)}\)} (y0m.south);
        \draw[->,green] (x0m.north) to node[right=5pt]  {\(\rotatebox{90}{\(>\)}\)} (y1m.south);
        \draw[->,green,out=180,in=-90] (y0m.north) to node[left = 5pt]  {\(\rotatebox{90}{\(>\)}\)} (z0m.south);
        \draw[->,green] (y0m.north) to node[above right = -3pt]  {\(\rotatebox{90}{\(>\)}\)} (z1m.south);
        \draw[->,green] (y1m.north) to node[above left = -3pt]  {\(\rotatebox{90}{\(>\)}\)} (z0m'.south);
        \draw[->,green,out=180,in=-90] (y1m.north) to node[right=11pt]  {\(\rotatebox{90}{\(>\)}\)} (z1m'.south);
        \draw[->,blue,dashed,out=180,in=-170] (zm) to node[below left] {\(x_0\)} (ym);
        \draw[->,blue,dashed,out=0,in=-10] (zm') to node[below right] {\(x_0\)} (ym);
    \end{tikzpicture}
    \end{minipage}

    \begin{align*}
        \phi_0(x_0)&\coloneqq\forall_{x_0',y_0',y_1':\bb N}\,(x_0'<x_0\to y_0'< x_0'\to y_1'<x_0'\to A(y_0',y_1')).
    \end{align*}
    This induction hypothesis can be applied at both buds, while induction on just \(y_0\) or \(y_1\) would not work.
    Here we deviate from the approach of Leigh \& Wehr \cite{leigh2025unravelling}, who solve such issues by doing induction on the {maximum} of \(y_0\) and \(y_1\), which does not always exist if \(<\) is not linear.
\end{example}

%%novalidate

\section{From a (well-organized) cyclic to an inductive proof}\label{sec:transformation}

In this section, we set out to prove our main theorem: 

\begin{theorem}\label{the:main}
     For every $\cal{C}$-proof $\pi$ of a judgment \(A\) with trace objects \(\overline{S}\), there exists a \(\cal C_{\sf{ind}}\)-proof $\pi_{\sf{ind}}$ of \(\overline{x:S}\vdash A(\overline x)\).
\end{theorem}
Let us fix a $\cal C$-proof $\pi$. Its translation to $\pi_{\sf{ind}}$ will factor through a cyclic representation of $\pi$. As illustrated in the previous section, not every cyclic representation is suitable for this.
We will define what it means for a cyclic representation to be \emph{well-organized}, and then show how any such representation can be translated to an inductive proof. In \autoref{sec:safra}, we show that any proof indeed has a well-organized cyclic representation.

\subsubsection{Defining a well-organized representation.} Before we can state our definition, we have to introduce some notions.
First, for any cyclic representation, write \(b'\triangleright b\) if \(\sf{sprout}(b)\) 

\vspace{2pt}
\noindent
\begin{minipage}[c]{0.8\textwidth}
    \setlength{\parindent}{1.7em}\noindent
    lies between \(\sf{sprout}(b')\) and \(b'\) in the tree (see the diagram on the right). Intuitively, \(b'\triangleright b\) means that $b$ corresponds to an `inner' induction of the inductive argument represented by $b'$, which in particular means that the induction hypothesis of $b$ will depend on the `outer' induction hypothesis of $b'$. As this dependency-relation on induction hypotheses is transitive, we will be particularly interested in the relation \(\triangleright^*\).
\end{minipage}\hfill\begin{minipage}[c]{0.2\textwidth}
    \centering
    \begin{tikzpicture}[scale=0.6]
        \node (b0)  at (0,  0) {\(b\)};
        \node (s0)  at (-1, -1) {\(\bullet\)};
        \node (b1)  at (-2,  0) {\(b'\)};
        \node (s1)  at (-2, -2) {\(\bullet\)};
        \path (s0)  edge (s1);
        \path (s0)  edge (b0);
        \path (s0)  edge (b1);
        \path [->, red, dashed, bend left] (b0) edge (s0);
        \path [->, blue, dashed] (b1) edge (s1);
    \end{tikzpicture}
    \smallskip
    \(b'\triangleright b\)
\end{minipage}
\vspace{-3.2pt}

Second, since the translation will essentially consist of `introducing an induction hypothesis' at sprouts and `using an induction hypothesis' at buds, 
it will be useful to make the bud and sprout nodes \emph{explicit} in the form of inference rules: \begin{definition}[explicit cyclic representation]
    An \emph{explicit cyclic representation} in \(\cal C\) is a finite derivation in the system \(\cal C\) extended with the rules: \begin{align*}
        \prfstack[r]{\(\rm{bud}\),}
            {A}\, &&
        \prfstack[r]{\(\rm{sprout}\),}
            {A}
            {A}
    \end{align*} together with a bijection \(\sf{sprout}\) that maps every node labeled by a bud rule to a node labeled by a sprout rule.
    An explicit cyclic derivation \(\pi_{\sf{cyc}}\) \emph{represents} a $\cal C$-derivation \(\pi\) if there exists a map \(\sf{rep}:\sf{Node}_{\pi_{\sf{cyc}}}\to\sf{Node}_\pi\) such that $\sf{rep}(\sf{root}_{\pi_{\sf{cyc}}})=\sf{root}_{\pi}$ and for every \(n\in\sf{Node}_{\pi_{\sf{cyc}}}\): \begin{description}
        \item[\(\cal C\)-rule.] if \(\sf{rule}(n)\) is a rule of \(\cal C\), then \(\sf{rep}(n)\) has the same rule and for every \(i\) we have \(\sf{rep}(\sf{Child}(n)_i)=\sf{Child}(\sf{rep}(n))_i\);
        \item[Bud.] if \(n\) is a bud, then \(\sf{rep}(n)\) and \(\sf{rep}(\sf{sprout}(n))\) have isomorphic subderivations and \(\sf{rep}(\sf{sprout}(n))\) lies strictly below \(\sf{rep}(n)\) in $\pi$;
        \item[Sprout.] if \(n\) is a sprout, then \(\sf{rep}(n)=\sf{rep}(\sf{Child}(n)_0)\).
    \end{description}
\end{definition}
Note that, in contrast to non-explicit ones, each sprout in an explicit cyclic representation has a unique bud. Clearly, any (non-explicit) representation can be extended to an explicit one; one just has to decide upon the order in which to treat buds with the same sprout.

Third, for any node \(n\) in \(\pi\), we fix a list of distinct fresh variables \(\sf{Var}(n)\) of size \(|\sf{Ob}(\sf{judg}(n))|\), which should be viewed as placeholders for the trace objects at $n$. Given an explicit representation \(\pi_{\sf{cyc}}\) of \(\pi\), for any \(n\in\sf{Node}_{\pi_{\sf{cyc}}}\) we will write: \begin{align*}
    A_n&\coloneqq\sf{judg}(\sf{rep}(n)), &
    r_n&\coloneqq\sf{rule}(\sf{rep}(n)), &
    \overline{S_n}&\coloneqq\sf{Ob}(A_n), &
    \overline{x_n}&\coloneqq\sf{Var}(\sf{rep}(n)).
\end{align*}
Moreover, if \(n\) is labelled by a \(\cal C\)-rule and \(n'\in\sf{Child}(n)\), we write \(G_{n,n'}:|\overline{S_n}|\to|\overline{S_{n'}}|\) for the corresponding size-change graph.
Note that any sprout \(s\) with child \(s'\) satisfies \(\sf{rep}(s)=\sf{rep}(s')\) and so in particular \(\overline{x_s}=\overline{x_{s'}}\). Lastly, we have to define the \emph{composition} of any two (strict) inequalities in $\cal C_{\sf{ind}}$: $x\geq y$ and $y\geq z$ compose to $x\geq z$, while ($x\geq y$ and $y>z$), ($x> y$ and $y\geq z$), and ($x> y$ and $y> z$) all compose to $x> z$.

We are now ready to state our main definition:
\begin{definition}[well-organized cyclic representation]\label{def: well org cyc rep}
A \emph{well-organized cyclic representation} is an explicit cyclic representation \(\pi_{\sf{cyc}}\) with the following additional structure:
    \begin{description}
        \item[Relevant variables and inequalities.]
        For every node $n$, we have a variable context \(\cal X_n\) and a formula context \(\cal Q_n\) of (strict) inequalities between variables in \(\cal X_n\) such that: 
        \begin{description}
            \item[Root.] $\cal X_{\sf{root}}=\overline{x_{\sf{root}}:S_{\sf{root}}}$ and \(\cal Q_{\sf{root}}=()\);
            \item[\(\cal C\)-rule.] if $n'$ is a child of a node \(n\) labeled by a \(\cal C\)-rule, then \(\cal X_{n'}=\cal X_{n},\overline{x_{n'}:S_{n'}}$ and \(\cal Q_{n'}\) is the closure of \(\cal Q_n,\overline{x_n}\gtrsim_{G_{n,n'}}\overline{x_{n'}}\) under composition of (strict) inequalities;
            \item[Sprout.] if $s'$ is the child of a sprout $s$, then $\overline{x_s:S_s}\subseteq\cal X_{s'}\subseteq \cal X_s$ and \(\cal Q_{s'}\subseteq\cal Q_s\).
        \end{description}
        \item[Progress variable.]
        For every bud $b$, we have a designated variable $x_{\sf{prog}_b}$ in \(\cal X_b\).
        \item[Substitution.] 
        For every sprout \(s\) with child $s'$ and bud $b$, we have a substitution \(\rho_b\) sending every variable \(x\) in \(\cal X_{s'}\) to a variable \(x[\rho_b]\) in \(\cal X_b\) of the same sort such that:
        \begin{enumerate}
            \item \((x_{\sf{prog}_{b\phantom'}}>x_{\sf{prog}_{b\phantom'}}[\rho_b])\in\cal Q_b\);
            \item\((x_{\sf{prog}_{b'}}\ge x_{\sf{prog}_{b'}}[\rho_b])\in\cal Q_b\) for every $b'\triangleright^*b$;
            \item \((y[\rho_b]> z[\rho_b])\in\cal Q_b\) for every \((y> z)\in \cal Q_{s'}\), and \\
            \((y[\rho_b]\ge z[\rho_b])\in\cal Q_b\) for every \((y\geq z)\in \cal Q_{s'}\);
            \item $x_{b,j}=x_{s,j}[\rho_b]$ for every trace object $j$.
        \end{enumerate}
    \end{description}
\end{definition}
Intuitively, the relevant variables and inequalities are those that we need to remember; in particular, at a sprout we forget what should not be included in the induction hypothesis.
The progress variable $x_{\sf{prog}_b}$ is the variable that we do induction on at $\sf{sprout}(b)$, and the substitution $\rho_b$ allows us to use the induction hypothesis at the bud $b$.

\subsubsection{Translation to an inductive proof.} Assume that $\pi_{\sf{cyc}}$ is a well-organized cyclic representation of $\pi$. The following proposition states how $\pi_{\sf{cyc}}$ can be translated into an inductive proof; we then obtain \autoref{the:main} as a direct consequence of it.

\begin{proposition}\label{prop:main}
    For every node \(n\) in \(\pi_{\sf{cyc}}\), there is a $\cal C_{\sf{ind}}$-proof of the sequent \begin{align*}
        \cal X_n;\cal{H}_{n},\cal{Q}_{n}\vdash A_n(\overline{x_n}),
    \end{align*} where $\cal{H}_{n}$ is the set of \emph{induction hypotheses} defined by upward induction on \(\pi_{\sf{cyc}}\): 
    \begin{description}
        \item[Root.] \(\cal H_{\sf{root}}\coloneqq()\);
        \item[\(\cal C\)-rule.] $\cal{H}_{n'}\coloneqq\cal{H}_{n}$, if $n'$ is a child of a node \(n\) labeled by a \(\cal C\)-rule;
        \item[Sprout.] $\cal{H}_{s'}\coloneqq(\cal H^{\triangleright^*}_{s},\phi_b(x_{\sf{prog}_b}))$ if $s'$ is the child of a sprout \(s\) with bud $b$, where 
        \begin{align*}
            \phi_b(x_{\sf{prog}_b})&\coloneqq\forall_{\cal X_{s'}'}\,(x_{\sf{prog}_b}>x'_{\sf{prog}_b}\to(
           \cal{H}^{\triangleright^*}_{s} \to\cal{Q}_{s'}\to
            A_{s}(\overline{x_s}))[\cal X_{s'}'/\cal X_{s'}^{}]),
        \end{align*}
           and \(\cal H_s^{\triangleright^*}\subseteq\cal H_s\) consists of those induction hypotheses \(\phi_{b'}(x_{\sf{prog}_{b'}})\) in \(\cal H_s\) with \(b'\triangleright^*b\).
    \end{description}
\end{proposition}
\begin{proof}
     For every sprout $s$ with child $s'$ and bud $b$, note that the induction hypothesis $\phi_b(x_{\sf{prog}_b})$ indeed has $x_{\sf{prog}_b}$ as its only free variable, as $x_{\sf{prog}_{b'}}$ occurs in $ \cal X_{s'}$ for any \(b'\triangleright^*b\). By downward induction on \(n\), we give an explicit construction of the required $\cal C_{\sf{ind}}$-proof: 
     \begin{description}
         \item[Bud.] For every bud \(b\) with sprout \(s\), the construction has the following form:\[
            \prftree[r]{bud}
                {A_b}
            \qquad\rightsquigarrow \qquad
            \prftree[r,double]{\(\forall/\to\)-elim.}
                {\prftree[r]{id}
                    {\cal X_b;\cal{H}_b,\cal{Q}_b\vdash\phi_b(x_{\sf{prog}_b})}}
                    {\prftree
                        {\cdots}}
                {\cal X_b;\cal{H}_b,\cal{Q}_b\vdash A_b(\overline{x_b})}
        \]
        The idea here is that we apply the induction hypothesis \begin{align*}
             \phi_b(x_{\sf{prog}_b})&\coloneqq\forall_{\cal X_{s'}'}\,(x_{\sf{prog}_b}>x'_{\sf{prog}_b}\to(
           \cal{H}^{\triangleright^*}_{s} \to\cal{Q}_{s'}\to
            A_{s}(\overline{x_s}))[\cal X_{s'}'/\cal X_{s'}^{}]),
        \end{align*} to the variables \(\cal X_{s'}[\rho_b]\) to deduce $A_{b}(\overline{x_b})$.
        Note that for every \(b'\triangleright^* b\) we have $\phi_{b'}(x_{\sf{prog}_{b'}})\in\cal{H}_{b}$, since every \(\sf{sprout}(b'')\) between \(s\) and \(b\) satisfies \(b'\triangleright^* b\triangleright b''\) and thus \(\phi_{b'}\) is not removed at \(\sf{sprout}(b'')\).
        So, in particular \(\phi_b(x_{\sf{prog}_b})\in\cal H_b\), and by property 4 of \(\rho_b\) we have $A_s(\overline{x_s}[\rho_b])=A_b(\overline{x_b})$, since buds and sprouts have the same judgment.
        Moreover, we see that the three antecedents can be derived from \(\cal H_b\) and \(\cal Q_b\) using the corresponding property of \(\rho_b\): \begin{enumerate}
            \item we derive \(x_{\sf{prog}_b}>x_{\sf{prog}_b}[\rho_b]\) since \((x_{\sf{prog}_b}>x_{\sf{prog}_b}[\rho_b])\in\cal Q_b\),
            \item we derive every \(\phi_{b'}(x_{\sf{prog}_{b'}}[\rho_b])\in\cal H_s^{\triangleright^*}[\rho_b]\) since \(\phi_{b'}(x_{\sf{prog}_{b'}})\in\cal H_b\) and \((x_{\sf{prog}_{b'}}\ge x_{\sf{prog}_{b'}}[\rho_b])\in\cal Q_b\), and the only free variable \(x_{\sf{prog}_{b'}}\) of \(\phi_{b'}(x_{\sf{prog}_{b'}})\) functions as an upper bound on its universal quantifier.
            \item we derive every \((y[\rho_b]\gtrsim z[\rho_b])\in\cal Q_{s'}[\rho_b]\) since it is contained in \(\cal Q_b\).
        \end{enumerate}
        
        \item[Sprout.] For every sprout \(s\) with child \(s'\) and bud \(b\), we do the following: \[
            \prftree[r]{sprout}
                {A_s}
                {A_s}
            \qquad\rightsquigarrow \qquad
            \prftree[r]{weaken.}
                {\prftree[r]{ind\(_{x_{\sf{prog}_b}}\)}
                    {\cal X_{s'};\cal{H}_s^{\triangleright^*},\phi_b(x_{\sf{prog}_b}),\cal{Q}_{s'}\vdash A_s(\overline{x_s})}
                    {\cal X_{s'};\cal{H}_s^{\triangleright^*},\cal{Q}_{s'}\vdash A_s(\overline{x_s})}}
                {\cal X_s;\cal{H}_s,\cal{Q}_s\vdash A_s(\overline{x_s})}
        \]
        Recall that $\cal{H}_{s'}=\cal H^{\triangleright^*}_{s},\phi_b(x_{\sf{prog}_b})$.
        By assumption, \(\cal X_{s'}\subseteq\cal X_s\) and \(\cal Q_{s'}\subseteq\cal Q_s\), and by definition \(\cal H_s^{\triangleright^*}\subseteq\cal H_s\), so weakening is justified.
        Since $x_{\sf{prog}_b}\in \cal X_{s'}$, the rule ind\(_{x_{\sf{prog}_b}}\) can be applied.

        \item[\(\cal C\)-rule.] For every node \(n\) with a \(\cal C\)-rule and child \(n'\) we do the following: \[
            \prfstack[r]{\(r_n\)}
                {A_{n'}}
                {A_n}
            \qquad\rightsquigarrow\qquad
            \prftree[r]{\(r_n\).}
                {\prftree[r,double]{cut}
                    {\cal X_n,\overline{x_{n'}:S_{n'}};\cal{H}_{n},\cal{Q}_{n'}\vdash A_{n'}(\overline{x_{n'}})}
                    {\prftree
                        {\cdots}}
                    {\cal X_n,\overline{x_{n'}:S_{n'}};\cal{H}_{n},\cal{Q}_n,\overline{x_n}\gtrsim_{G_{n,n'}}\overline{x_{n'}}\vdash A_{n'}(\overline{x_{n'}})}}
                {\cal X_n;\cal{H}_{n},\cal{Q}_{n}\vdash A_n(\overline{x_n})}
        \]
        Recall that \(\cal X_{n'}=\cal X_n,\overline{x_{n'}:S_{n'}}\) and \(\cal H_{n'}=\cal H_n\). Note that the additional inequalities in \(\cal Q_{n'}\) are all compositions of inequalities in \(\cal{Q}_n,\overline{x_n}\gtrsim_{G_{n,n'}}\overline{x_{n'}}\), each of which we can cut in since they are already provable from the assumptions due to the trans-rule.
        \end{description}
\end{proof}
\begin{proof}[of \autoref{the:main}]
    This follows from \autoref{prop:main} since \(\cal H_{\sf{root}}=()\) and \(\cal Q_{\sf{root}}=()\).
\end{proof}
\begin{remark}
    Let us emphasize here that the proof of \autoref{prop:main} is fully constructive, and that the obtained inductive proof $\pi_{\sf{ind}}$ structurally remains very close to the well-organized cyclic representation $\pi_{\sf{cyc}}$ (and thereby to a finite subtree of the proof $\pi$). In particular, the structure of the $\cal C$-rules in $\pi_{\sf{cyc}}$ remains exactly the same: to obtain $\pi_{\sf{ind}}$, we essentially keep every $\cal C$-rule application of $\pi_{\sf{cyc}}$ in place, while we add some weakening and an induction rule at sprouts, and we add some first-order (that is, $\cal C$-rule and induction free) reasoning at the buds and the nodes labelled by a $\cal C$-rule.
\end{remark}

%%novalidate

\newcommand{\hcancel}[1]{%
    \tikz[baseline=(tocancel.base)]{
        \node[red, inner sep=0pt,outer sep=0pt] (tocancel) {\(\vphantom|\smash{#1}\)};
        \draw[red, line width = 0.75pt] (tocancel.south west) -- (tocancel.north east);
    }%
}%

\section{Existence of a well-organized cyclic representation}\label{sec:safra}
Fix a (possibly infinite) $\cal{C}$-proof $\pi$.
To find a well-organized cyclic representation, we use a well-known strategy in cyclic proof theory, which goes back to Safra's construction for determinizing B\"uchi automata \cite{SafraConstruction}: we annotate sequents with stacks of names to obtain `reset proofs' \cite{jungteerapanich2010tableau,Stirling,leigh_gtc_2024}. 
We follow the approach of Wehr \cite{wehr2025cyclic}, with one important change: we require a name to be `uniformly covered' rather than just `covered' before we allow a reset.
We then follow Sprenger \& Dam \cite{sprenger_structure_2003} to unfold any reset proof into one where the induction order of the cycles respects the relation $\triangleright^*$. 

\subsubsection{Reset proofs.}
We will assign some additional data to every node in $\pi$, referred to as its \emph{annotations}.
These annotations will be used to determine when we can introduce a cycle; namely, when a node and some lower node have the same subderivation and annotations.
The annotations of any node in $\pi$ will be determined by the annotations of its parent.
Given a path \(n_0,\dots,n_k\) from the root, for any $i\le k$ we will abbreviate:
\begin{align*}
    A_i\defeq\sf{judg}(n_i),\qquad\overline{x_i}\defeq\sf{Var}(n_i),\qquad  G_i\defeq\sf{Graph}(\sf{rule}(n_{i}))_{l_i},
\end{align*} when this does not lead to confusion; here \(G_i\) is only defined for \(i\le k-1\) and \(l_i\) is the index such that $n_{i+1}=\sf{Child}(n_i)_{l_i}$.
We call \(x_{i,j}\) an \emph{ancestor} of \(x_{i',j'}\) iff \(i\le i'\) and the index \(j\) at \(n_i\) is connected to the index \(j'\) at \(n_{i'}\) by a path through the size-change graphs \(G_i,\dots,G_{i'-1}\); it is a \emph{strict ancestor} if there is a progressing edge on this path.  
Moreover, if \((i,j)\) is (strictly) earlier in the lexicographical ordering than \((i',j')\), then we call \(x_{i,j}\) \emph{(strictly) older} than \(x_{i',j'}\).
Note that any (strict) ancestor is necessarily (strictly) older, while the converse need not hold.

The annotations of the node \(n_k\) will keep track of `relevant' ancestors of the variables \(\overline{x_k}\) (namely the oldest one, and some of the ones that made progress).
To keep the number of possible annotations finite, throughout the possibly infinite tree $\pi$ we will only use finitely many names from some fixed set \(\sf{Name}\coloneqq\{a_0,a_1,\dots\}\) to refer to ancestors. 
\begin{definition}[annotations]\label{def:safra}
    For every path \(n_0,\dots,n_k\) from the root in $\pi$, we define: \begin{description}
        \item[Currently assigned names.] a non-repeating list \(\sf{Name}(n_k)\in\sf{List}(\sf{Name})\);
        \item[Name to variable assignment.] for each \(a\in\sf{Name}(n_k)\) a variable \(\sf{var}(n_k,a)\in\sf{Var}(n_k)\);
        \item[Ancestors of \(\overline{x_{k}}\).] for each $j<|\sf{Ob}(A_k)|$ a sublist \(\sf{Stack}(n_k)_j\) of \(\sf{Name}(n_k)\).
    \end{description}
    \(\sf{Name}(n_k)\) will be ordered from old to young corresponding variables.
    We use induction: \begin{itemize}
        \item At the root \(n_0\) there is no prior history, so we give every \(x_{0,j}\) a new name: we define \(\sf{Name}(n_0)\coloneqq(a_0,\dots,a_{|\sf{Ob}(A_0)|-1})\), \(\sf{var}(n_0,a_j)\coloneqq x_{0,j}\), and \(\sf{Stack}(n_0)_j\coloneqq(a_j)\).
        \item 
        For \(n_{k+1}\), we first introduce new names for the new variables \(\overline{x_{k+1}}\), then we pick the most important stack of ancestors from a number of candidates, and lastly we perform a number of `resets' to forget unneeded information:
        \begin{description}
            \item[New names.]
            Let \(a'_{0},a'_1,\dots\) be the first names not in the list \(\sf{Name}(n_k)\) and assign: \[
                \sf{Name}(n_{k+1})\coloneqq(\sf{Name}(n_k),a'_0,\dots,a'_{|\sf{Ob}(A_{k+1})|-1}).
            \]
            We assign these new names to the new variables \(\sf{var}(n_{k+1},a'_{j'})\coloneqq x_{k+1,j'}\), and copy the assignment \(\sf{var}(n_{k+1},a_j)\defeq\sf{var}(n_k,a_j)\) for \(a_j\) already in \(\sf{Name}(n_k)\).
            \item[New stacks.]
            Every edge from some \(j\) to \(j'\) in \(G_k\) generates a candidate stack for \(\sf{Stack}(n_{k+1})_{j'}\), namely: the list \((\sf{Stack}(n_k)_j,a'_{j'})\) if the edge is progressing, and the list \(\sf{Stack}(n_k)_j\) otherwise.
            We assign to \(\sf{Stack}(n_{k+1})_{j'}\) the `oldest' such candidate, where a stack is \emph{older} than another stack if it contains the oldest name in the symmetric difference \(S\mathbin{\triangle} S'\coloneqq(S\setminus S')\cup(S'\setminus S)\) of their sets of elements \(S\) and \(S'\).
            If there are no candidates, we start a new stack: \(\sf{Stack}(n_{k+1})_{j'}\coloneqq(a_{j'}')\).
            \item[Reset.] Check whether there is a name $a$ in that is \emph{uniformly covered} by another $a'$, meaning that \(a\) appears in some \(\sf{Stack}(n_{k+1})_{j'}\), and: 
            \begin{itemize}
                \item \(a\) appears strictly before \(a'\) in the list \(\sf{Name}(n_{k+1})\);
                \item for every \(j'\), if $a$ occurs in \(\sf{Stack}(n_{k+1})_{j'}\), so does $a'$.
            \end{itemize}
            If so, take the youngest such \(a\), and perform a \emph{reset on $a$}: for every \(j'\), we reassign \(
                \sf{Stack}(n_{k+1})_{j'}\coloneqq\sf{Stack}(n_{k+1})_{j'}\upharpoonright a,
            \) where \(S\upharpoonright a\) is the prefix of \(S\) up to and including \(a\) if \(a\) appears in \(S\), and the entire stack otherwise.
            Keep applying resets until no name is uniformly covered.
            Then, remove all names that do not appear on one of the stacks from \(\sf{Name}(n_{k+1})\), and record the names that have been reset as \(\sf{Reset}(n_{k+1})\subseteq\{a_0,\dots\}\).
        \end{description}
    \end{itemize}
\end{definition}
For some examples of annotated proofs, see Appendix \ref{app:reset}.
Note that, when we follow a branch $n_0,n_1,\dots$ through $\pi$, the lists $\sf{Stack}(n_k)_j$ are indeed treated like stacks: new names are added to the right and we delete a suffix when doing a reset.

With these annotations, we can obtain a cyclic representation in which buds and sprouts have the same relevant history and in which every bud $b$ can be assigned a trace object that has progressed between $\sf{sprout}(b)$ and $b$.
\begin{theorem}[existence of a reset proof]\label{thm:existence cyclic rep with prog name}
        Every $\cal{C}$-proof has a cyclic representation in which every branch $n_0,\dots, n_m,\dots ,n_k$ with \(n_k\in\sf{Bud}\) and \(\sf{sprout}(n_k)=n_m\) satisfies:
    \begin{enumerate}
        \item $\sf{Name}(n_m)=\sf{Name}(n_k)$ and for all trace objects \(j\): $\sf{Stack}(n_m)_j=\sf{Stack}(n_k)_j$;
        \item there exists a name $\sf{prog}(n_k)\in\sf{Reset}(n_m)\cap\sf{Reset}(n_k)$ (the \emph{progressing name}) that occurs in \(\sf{Name}(n_i)\) for every \(m\le i\le k\).
    \end{enumerate}
\end{theorem}
Note that $\sf{Name}(n_m)=\sf{Name}(n_k)$ and $\sf{prog}(n_k)\in\bigcap_{m\le i\le k}\sf{Name}(n_i)$ imply that the stack $\sf{Name}(n_k)\upharpoonright\sf{prog}(n_k)$ (often called the \emph{invariant} of $n_k$) is a prefix of $\sf{Name}(n_i)$ for every \(m\le i\le k\), because new names can only be added to the right of \(\sf{Name}(n_i)\).
This means that every name $a$ in $\sf{Name}(n_m)$ older than $\sf{prog}(n_k)$ is preserved from the sprout $n_m$ to the bud $n_k$, so \(\sf{var}(n_i,a)\) denotes the same variable for every \(m\le i\le k\).

To prove \autoref{thm:existence cyclic rep with prog name}, we first prove two lemmas.
The first states that we have a bound on the size of our annotations, while the second lemma states that every infinite branch of $\pi$ eventually contains a name that is preserved and progresses infinitely often:
\begin{lemma}\label{lem: finitely many names}
   For any node $n_k$ in $\pi$ with path $n_0,\dots,n_k$ from the root we have $\sf{Name}(n_k)\in\{a_0,\dots,a_{2^m+m-1}\}^*$, where $m\coloneqq\mathsf{max}\{|\sf{Ob}(A_i)|:i\le k\}$.
\end{lemma}
\begin{proof}
    If two names occur in exactly the same stacks in $(\sf{Stack}(n_i)_j)_{j<|\sf{Ob}(A_i)|}$, then one of them must be uniformly covered, which would have triggered a reset.
    Thus, we can use at most $2^m$ names.
    However, it is possible that \(\sf{Name}(n_i)\) skips certain names because we first add new names before possibly removing names with a reset.
    By induction on \(k\), we see that $\sf{Name}(n_k)\subseteq\{a_0,\dots,a_{2^m+m-1}\}$ and \(|\sf{Name}(n_k)|\le 2^m\), where the induction step uses the fact that we always pick the first available names.
\end{proof}
\begin{lemma}\label{lem: existence preserved progress name} 
    For every infinite branch \((n_i)_{i<\omega}\) of $\pi$, there is an index $k$ and some name $a$ that occurs in $\sf{Name}(n_i)$ for every $i\ge k$ and in $\sf{Reset}(n_i)$ for infinitely many $i\ge k$.
\end{lemma}
\begin{proof}[sketch]
    As $\pi$ is a proof, this branch must contain a progressing trace of the form $(j_i)_{i\geq l}$ with $j_i<|\sf{Ob}(A_i)|$. It can be shown that the corresponding sequence of stacks $(\sf{Stack}(n_i)_{j_i})_{i\ge l}$ has a prefix that eventually becomes stable and whose last name is reset infinitely often.
    See Appendix \ref{app: proof of lem exist progress name} for a detailed proof. 
\end{proof}

\begin{proof}[of \autoref{thm:existence cyclic rep with prog name}]
    For any infinite branch \((n_i)_{i<\omega}\), we show that we can introduce a back-edge that satisfies the requirements of the theorem; the theorem then follows by K\"onig's lemma.
    Let $l$ and $a$ be given by \autoref{lem: existence preserved progress name}.
    By regularity, \(\pi\) only has finitely many subderivations up to isomorphism, and so by \autoref{lem: finitely many names} there are only finitely many options for \(\sf{Name}(n_i)\) and \((\sf{Stack}(n_i)_j)_j\).
    Moreover, there are infinitely many $i\ge l$ for which $a\in\sf{Reset}(n_i)$, so by a pigeonhole argument we can find \(l\le m<k\) such that \(n_m\) and \(n_k\) have isomorphic subderivations and satisfy the requirements of the theorem.
    
\end{proof}

\subsubsection{Induction order.}
Recall from \autoref{exam:importorder} that the order in which we introduce induction hypotheses matters.
The annotations allow us to formalize this `induction order'. A bud \(b'\) is \emph{more important} than a bud $b$ if $\sf{Name}(b')\upharpoonright \sf{prog}(b')$ is a prefix of $\sf{Name}(b)\upharpoonright \sf{prog}(b)$ and for every $a$ in $\sf{Name}(b')\upharpoonright\sf{prog}(b')$, we have $\sf{var}(b,a)=\sf{var}(b',a)$. Intuitively, this means the progress variable of $b'$ and its relevant history are preserved by $b$. 

To obtain a well-organized representation, we need the importance order of the buds to match the order $\triangleright^*$ that describes how buds and their sprouts occur in the tree. 
\begin{proposition}[existence of a reset proof respecting the induction order]\label{prop:SD transitive}
    Every $\cal{C}$-proof has a cyclic representation with progress map $\sf{prog}$ that satisfies the properties of \autoref{thm:existence cyclic rep with prog name} and in addition: if $b'\triangleright^*b$  then $b'$ is more important than $b$.
\end{proposition}
\begin{proof}[sketch]
    We follow Sprenger \& Dam \cite[Theorem~5]{sprenger_structure_2003}: starting from the cyclic representation given by \autoref{thm:existence cyclic rep with prog name}, we keep unfolding the buds of this representation while only introducing a back-edge for a (copy of a) bud if all more important buds have already been assigned a sprout lower in the tree. See Appendix \ref{app: proof of prop SD trans}.
\end{proof}

\subsubsection{Well-organized cyclic representation.} We now show how the cyclic representation $\pi_{\sf{cyc}}$ given by \autoref{prop:SD transitive} induces a well-organized cyclic representation of $\pi$. 

First, we turn $\pi_{\sf{cyc}}$ into an explicit cyclic representation $\pi'_{\sf{cyc}}$ of $\pi$: to any bud of $\pi_{\sf{cyc}}$ we add an application of the bud-rule, and to any sprout $s$ of $\pi_{\sf{cyc}}$ we add one application of the sprout-rule per bud of $s$.
The latter results in a `copy' of $s$ for each of its buds; we then define a new bijective sprout map $\sf{sprout}'$ that assigns a unique such copy to each bud $b$ of $s$ in a way that respects the importance order: for any $b,b'\in\sf{sprout}^{-1}(s)$ we have $b'\triangleright^*b$ in $\pi_{\sf{cyc}}$ iff $\sf{sprout}'(b')$ lies below $\sf{sprout}'(b)$ in $\pi'_{\sf{cyc}}$. 
Let $\sf{rep}:\sf{RepNode}'\to\sf{Node}$ be the map that sends any node in $\pi'_{\sf{cyc}}$ to the node in $\pi$ that it represents: $\sf{rep}(n)=s$ if $n$ is a copy of $s$, and $\sf{rep}(n)=n$ otherwise.
Via \(\sf{rep}\), nodes of $\pi'_{\sf{cyc}}$ inherit annotations of nodes in $\pi$: we write e.g. $\sf{Name}(n)$ for $\sf{Name}(\sf{rep}(n))$.

Second, we inductively define the required additional data for $\pi'_{\sf{cyc}}$: 
\begin{description}[font=\itshape]
        \item[Relevant variables.] For every node $n$ in $\pi'_{\sf{cyc}}$, the contexts \(\cal X_n\) and \(\cal Q_n\) are given by: 
        \begin{description}
            \item[Root.] We set $\cal X_{\sf{root}}\defeq\overline{x_{\sf{root}}:S_{\sf{root}}}$ and \(\cal Q_{\sf{root}}=()\).
            \item[\(\cal C\)-rule.] If $n'$ is a child of a node \(n\) labelled by a \(\cal C\)-rule, then \(\cal X_{n'}=\cal X_{n},\overline{x_{n'}:S_{n'}}$ and \(\cal Q_{n'}\) is the closure of \(\cal Q_n,\overline{x_n}\gtrsim_{G_{n,n'}}\overline{x_{n'}}\) under composition of (strict) inequalities;
            \item[Sprout.] If $s'$ is the (unique) child of a sprout $s$, then $\cal X_{s'}$ is the subset of $\cal X_{s}$ that contains the variables $\overline{x_s}$ plus all variables $x$ that satisfy $x=\sf{var}(s,a_x)$ for some name $a_x\in\sf{Name}(s)$. Moreover, $\cal Q_{s'}$ is the subset of $\cal Q_{s}$ that contains all inequalities of the form $\sf{var}(s,a)\geq x_{s,j}$ with $a$ the topmost name in $\sf{Stack}(s)_j$ and of the form $\sf{var}(s,a)> \sf{var}(s,a')$ with $a,a' \in\sf{Stack}(s)_j$ for some $j$. 
        \end{description}
        \item[Progress variable.]
        For every bud $b$ we define \(x_{\sf{prog}_b}\defeq\sf{var}(b,\sf{prog}(b))\).
        \item[Substitution.] For every sprout $s$ with child $s'$ and bud $b$, let $\sf{cov}(b)$ be a name that uniformly covered $\sf{prog}(b)$ before performing the reset on $\sf{prog}(b)$ at $b$. Since $\sf{Name}(s)=\sf{Name}(b)$, we may define the substitution $\rho_b$ on variables in $\cal X_{s'}$ by:
        \begin{align*}
        x[\rho_b]\coloneqq \begin{cases}
            x_{b,j} & \text{if }x=x_{s,j}\text{ for some trace object }j; \\
           \sf{var}(b,\sf{cov}(b)) &\text{if }x=x_{\sf{prog}_b}; \\
            \sf{var}(b,a_x) & \text{otherwise}.
        \end{cases}
        \end{align*}
    \end{description}
See Appendix \ref{app: proof that the renaming works} for a proof that the substitution $\rho_b$ is well-defined and indeed satisfies the requirements of \autoref{def: well org cyc rep}. We then obtain the main result of this section:
\begin{theorem}
    Every $\cal{C}$-proof has a well-organized cyclic representation.
\end{theorem}

\section{Conclusion}\label{sec:conclusion}
We extend the results of Wehr et al. \cite{afshari_abstract_2022,leigh_gtc_2024,leigh2025unravelling,wehr2025cyclic} to obtain a direct translation from cyclic proofs to inductive proofs in an abstract setting (\autoref{the:main}). 
To do so, we modify their reset proofs in one important way: instead of resetting a name whenever it is \emph{covered} by a set of names, we only do so once it is \emph{uniformly covered} by a single name.
This allows us to omit the assumption that the well-founded relation is linear (see also \autoref{exam:importforget}).

Our abstract notion of cyclic proofs covers a wide range of cyclic proofs in the literature.
This does not mean that our translation always provides an inductive proof in the `intended' inductive system; for the common cyclic systems for the modal $\mu$-calculus, the induced inductive system will be quite different from Kozen's inductive system \cite{kozen1983results}.
Nevertheless, our result reduces the finitization problem to finding a reduction from our induced inductive system to the intended one.
We expect this to be helpful for systems that extends first-order logic, such as \(k\)th-order logic as illustrated for \(\sf{CHA}k\) and \(\sf{HA}k\) in Appendix \ref{app:flattening}, as well as dependent type theory.

\paragraph{Future work.}
Our abstract framework can be used for cyclic proofs with both inductive and coinductive trace sorts (see \autoref{remark:parity}).
However, we always translate to a system with only induction rules and no {coinduction} rules.
This is similar to cyclic systems with \emph{ordinal approximations}, in which both inductive and coinductive reasoning can be reduced to an induction on ordinals; Sprenger \& Dam's system for the first-order \(\mu\)-calculus \cite{sprenger_structure_2003,SprengerDamComp} is an example of this, and so is the use of sized types in the proof assistant \texttt{Agda} \cite{abel2016well}.
It would be interesting to extend our translation to induced systems with explicit coinduction principles, or more generally, to a combination of inductive and coinductive principles \cite{danielsson2009mixing,basold2018mixed}.
This would require generalizing our notion of cyclic proof to allow for both a primitive notion of inductive and of coinductive traces;
it is known that the soundness/termination conditions for such traces can also be described using size-change graphs \cite{hyvernat_size-change_2025}.

Our results can be applied to constructive theories, and the translation from a well-organized cyclic representation to an inductive proof is constructive. Nevertheless, we have used a classical metatheory to show that every cyclic proof has a well-organized cyclic representation (K\"onigs lemma).
We believe the latter can also be shown constructively, as long as we pick the right reformulation of a proof: instead of a regular infinite derivation where every infinite branch to makes progress infinitely often, we take a cyclic derivation where every cycle makes progress \cite[Chapter 8]{wehr2025cyclic}, which is equivalent in a classical metatheory \cite{lee_size-change_2001}. 
Such a constructive proof would require reformulating Section~\ref{sec:safra} using constructive Ramsey theory \cite{hutchison_stop_2012}.

Lastly, we want to apply our result to show that pattern matching with size-change termination is conservative over primitive elimination rules, which would extend the conservativity results for pattern matching with structural recursion \cite{EliminatingDependentPatternMatching,cockx_pattern_2014}.
Our abstract framework is general enough to handle arbitrary inductive types, since every (indexed) inductive type has a well-founded `structural order' principle \cite{abel2002predicative}\cite[section~3]{hutchison_few_2006}\cite[subsection~2.2.2 and Lemma~4.8]{cockx_dependent_2017}, but more work is needed to clarify in which ways computational behaviour is preserved.
On the dual side, we want to investigate whether copattern matching \cite{Thibodeau,hutchison_unnesting_2014,CockxAbel2020} on coinductive types is conservative over primitive introduction rules.

% \begin{credits}
%     \subsubsection{\ackname}
%     We are thankful for valuable perspectives and ideas from both the cyclic proof theory and dependent proof-assistant community; we are particularly grateful for insightful conversations with Andreas Abel, Bahareh Afshari, Benno van den Berg, Jesper Cockx, Gianluca Curzi, Anupam Das, Andr\'as Kov\'acs, Johannes Kloibhofer, Graham Leigh, Meven Lennon-Bertrand, Yde Venema, and Dominik Wehr.

%     % \subsubsection{\discintname}
%     % It is now necessary to declare any competing interests or to specifically
%     % state that the authors have no competing interests. Please place the
%     % statement with a bold run-in heading in small font size beneath the
%     % (optional) acknowledgments\footnote{If EquinOCS, our proceedings submission
%     % system, is used, then the disclaimer can be provided directly in the system.},
%     % for example: The authors have no competing interests to declare that are
%     % relevant to the content of this article. Or: Author A has received research
%     % grants from Company W. Author B has received a speaker honorarium from
%     % Company X and owns stock in Company Y. Author C is a member of committee Z.
% \end{credits}

\bibliographystyle{splncs04}
\bibliography{references,zotero}

\appendix

\section{(Cyclic) \texorpdfstring{$k$th}{kth}-order Heyting arithmetic}
\subsection{Axioms of \texorpdfstring{\((\sf C)\sf{HA}k\)}{(C)HAk}}\label{app: rules of CHAk}
In \((\sf C)\sf{HA}k\) we have the following standard axioms for the sort \(\bb N\): \begin{align*}
    \begin{aligned}
        \forall_{y:\bb N}\,(0+y&=y), \\
        \forall_{y:\bb N}\,(0\makesize\times{{}+{}}y&=0),
    \end{aligned} && \begin{aligned}
        \forall_{x,y:\bb N}\,\sf{suc}(x)+y&=\makesize{\sf{suc}}{y+{}}(x+y), \\
        \forall_{x,y:\bb N}\,\sf{suc}(x)\makesize\times{{}+{}} y&=y+(x\makesize\times{{}+{}}y),
    \end{aligned} && \begin{gathered}
        \forall_{y:\bb N}\,0\ne\sf{suc}(y), \\
        \forall_{x,y:\bb N}\,\sf{suc}(x)=\sf{suc}(y)\to x=y.
    \end{gathered}
\end{align*}
In addition, we have axioms for the powerset sorts.  For every variable context \(\cal X\), and formula \(\psi\) in \(\cal X,x:\cal P^i(\bb N)\) we have the axioms:
\begin{gather*}
    \exists_{X:\cal P(\cal P^i(\bb N))}\,\forall_{x:\cal P^i(\bb N)}\,(x\in X\liff \psi), \tag{specification} \\
    \forall_{X,X':\cal P(\cal P^i(\bb N))}\,(\forall_{x:\cal P^i(\bb N)}\,(x\in X\liff x\in X')\to X=X'). \tag{extensionality}
\end{gather*}
\subsection{Flattening proofs of \texorpdfstring{\((\sf{CHA}k)_\ind\)}{(CHAk)_ind} to \texorpdfstring{\(\sf{HA}k\)}{HAk}}\label{app:flattening}
In \((\sf{CHA}k)_\ind\) we have `outer' sequents, which are the judgments of \((\sf{CHA}k)_\ind\), and `inner' sequents, which are the judgments of \(\sf{HA}k\) that appear as relation symbols inside an outer sequent.
To differentiate the two, we write variables in outer sequents using \(\bb{blackboard}\) \(\bb{bold}\).
Given a formula \(\bb\delta\) of \((\sf{CHA}k)_\ind\), we write \(\lfloor\bb\delta\rfloor\) for the formula of \(\sf{HA}k\) obtained by replacing every atom \(A(\overline{\bb a })\) of \((\sf{CHA}k)_\ind\), where \(A\) is a sequent of the form \(\overline{x:\bb N},{\cal X};\Gamma\vdash\delta\) of \(\sf{CHA}k\), by the formula \(\forall_{\cal X}\,(\Gamma[\overline{\bb a/x}]\to\delta[\overline{\bb a/x}])\), and every \(\bb a<\bb b\) by \(\exists_{x:\bb N}\,(\sf{suc}(\bb a+x)=\bb b)\).
We show that every proof of \(\bb X;\bb\Gamma\vdash\bb\delta\) in \((\sf{CHA}k)_\ind\) can be turned into a proof of \(\bb X;\lfloor\bb\Gamma\rfloor\vdash\lfloor\bb\delta\rfloor\) in \(\sf{HA}k\).
In particular, given a proof of \(\overline{\bb x:\bb N}\vdash A(\overline{\bb x})\) in \((\sf{CHA}k)_\ind\) where \(A\) is a sequent \(\overline{x:\bb N},\cal X;\Gamma\vdash\delta\) of \(\sf{CHA}k\), this produces a proof of \(\overline{\bb x:\bb N}\vdash\forall_{\cal X}\,(\Gamma[\overline{\bb x/x}]\to\delta[\overline{\bb x/x}])\) in \(\sf{HA}k\), and using weakening and \({\to}/{\forall}\)-elim, we can turn this into a proof of \(\overline{x:\bb N},\cal X;\Gamma\vdash\delta\) in \(\sf{HA}k\).

To prove this, we show that applying \(\lfloor{-}\rfloor\) sends every rule instance of \((\sf{CHA}k)_\ind\) to a derivable rule instance of \(\sf{HA}k\).
There are three cases to consider: \begin{description}
    \item[\(\sf{CHA}k\)-rule.]
    If the \((\sf{CHA}k)_\ind\) rule instance comes from \(\sf{CHA}k\), then there is a \(\sf{CHA}k\) rule instance \[
        \prfstack[r]{$r$,}
            {A_i' & \qquad \text{for every $i<n$}}
            {A} \tag{\(\sf{CHA}k\)}
    \] (with corresponding trace objects \(\overline S\) and \(\overline{S_i'}\), and size-change graph \(G_i\), for $i<n$), such that the \((\sf{CHA}k)_\ind\) rule instance is of the form \[
         \prfstack[r]{$r$.}
            {\bb X,\overline{\bb x:S},\overline{\bb y:S_i'};\bb\Gamma,\overline{\bb x}\gtrsim_{G_i}\overline{\bb y}\vdash A_i'(\overline{\bb y}) & \qquad \text{for every $i<n$}}
            {\bb X,\overline{\bb x:S};\bb\Gamma\vdash A(\overline{\bb x})} \tag{\((\sf{CHA}k)_\ind\)}
    \]
    We have to show that applying \(\lfloor{-}\rfloor\) to all instances of the \((\sf{CHA}k)_\ind\)-rule gives a derivable rule instance of \(\sf{HA}k\): \[
         \prfstack[r]{?}
            {\bb X,\overline{\bb x:S},\overline{\bb y:S_i'};\lfloor\bb\Gamma,\overline{\bb x}\gtrsim_{G_i}\overline{\bb y}\rfloor\vdash\lfloor A_i'(\overline{\bb y})\rfloor & \qquad \text{for every $i<n$}}
            {\bb X,\overline{\bb x:S};\lfloor\bb\Gamma\rfloor\vdash\lfloor A(\overline{\bb x})\rfloor} \tag{\(\sf{HA}k\)}
    \]
    Here we have a few subcases: \begin{itemize}
        \item The most difficult case is where \(r\) is an instance of the \(\sf{case}_x\) rule in \(\sf{CHA}k\) \[\prfstack[r]{case\(_x\),}
            { {\cal X}; \Gamma[0/x]\vdash \delta[0/x]}
            { {\cal X},  x:\bb N; \Gamma[\sf{suc}(  x)/  x]\vdash \delta[\sf{suc}(  x)/  x]}
            { {\cal X},  x:\bb N; \Gamma\vdash \delta}
        \] which gives a rule of \((\sf{CHA}k)_\ind\) of the form \[
            \prfstack[r]{case\(_x\).}
                {\bb X,\bb x:\bb N;\bb\Gamma\vdash({\cal X}; \Gamma[0/x]\vdash \delta[0/x])}
                {\bb X,\bb x:\bb N,\bb y:\bb N;\bb\Gamma,\bb x>\bb y\vdash({\cal X},  x:\bb N; \Gamma[\sf{suc}(  x)/  x]\vdash \delta[\sf{suc}(  x)/  x])(\bb y)}
                {\bb X,\bb x:\bb N;\bb\Gamma\vdash({\cal X},  x:\bb N; \Gamma\vdash \delta)(\bb x)} \tag{\((\sf{CHA}k)_\ind\)}
        \] For notational convenience, we have assumed here that \(x:\bb N\) is the only variable of sort \(\bb N\) in \(\cal X,x:\bb N\) so that \(({\cal X}; \Gamma[0/x]\vdash \delta[0/x])\) is seen as a nullary relation symbol and \(({\cal X}; \Gamma[\suc(x)/x]\vdash \delta[\suc(x)/x])\) is seen as a unary relation symbol.
        Our goal is then to prove that the following is derivable: \[
        {\prfstack[r]{case\(_x\).}
                {\bb X,\bb x:\bb N;\lfloor\bb\Gamma\rfloor\vdash\forall_{\cal X}\,(\Gamma[0/x]\to\delta[0/x])}
                {\bb X,\bb x:\bb N,\bb y:\bb N;\lfloor\bb\Gamma,\bb x>\bb y\rfloor\vdash\forall_{\cal X}\,( \Gamma[\sf{suc}(\bb y)/  x]\to \delta[\sf{suc}(\bb y)/  x])}
                {\bb X,\bb x:\bb N;\lfloor\bb\Gamma\rfloor\vdash\forall_{\cal X}\,(\Gamma[\bb x/x]\to \delta[\bb x/x])}} \tag{\(\sf{HA}k\)}
        \]
        We will do this by reasoning informally in \(\sf{HA}k\).
        So, we first assume the two premises.
        Now, since we can do induction in \(\sf{HA}k\), we can do a case distinction to break the problem of proving the conclusion into two parts: proving the sequent for \(\bb x=0\), and proving the sequent for \(\bb x=\sf{suc}(\bb{x'})\) for some fresh variable \(\bb x':\bb N\).
        The first case is exactly the first assumption, and in the second case, since our second assumption applies to all \(\bb y\) such that \(\bb x>\bb y\), it certainly applies to \({\bb x}'\).
        \item If \(r\) is an axiom or an instance of a natural deduction rule in \(\sf{CHA}k\), then the corresponding size-change graphs only contain preserving paths from a variable to itself.
        This means that the rule in \((\sf{CHA}k)_\ind\) has the following form: \[
             \prfstack[r]{$r$.}
                {\bb X,\overline{\bb x:S},\overline{\bb y:S_i'};\bb\Gamma,\overline{\bb x\ge\bb y}\vdash A_i'(\overline{\bb y}) & \qquad \text{for every $i<n$}}
                {\bb X,\overline{\bb x:S};\bb\Gamma\vdash A(\overline{\bb x})} \tag{\((\sf{CHA}k)_\ind\)}
        \]
        Again, we reason informally in \(\sf{HA}k\) to show that applying \(\lfloor{-}\rfloor\) to this rule instance gives a derivable rule.
        Since \(\sf{HA}k\) has the same axioms and rules, it can show the implication \(\lfloor A_0'(\overline{\bb x})\rfloor\land\cdots\land\lfloor A_{n-1}'(\overline{\bb x})\rfloor\to\lfloor A(\overline{\bb x})\rfloor\).
        So, assume the \(n\) premises, and note that they apply to any \(\overline{\bb y}\) such that \(\overline{\bb x\ge\bb y}\), which means in particular that they apply to \(\overline{\bb y}=\overline{\bb x}\), so we can conclude \(\lfloor A(\overline{\bb x})\rfloor\) from our earlier implication.
    \end{itemize}
    \item[Intuitionistic multi-sorted predicate logic.]
    If we have a natural deduction rule that was newly added in \((\sf{CHA}k)_\ind\) (not induced by a natural deduction rule in \(\sf{CHA}k\), which is the case we just covered above), then \(\lfloor{-}\rfloor\) only replaces relation symbols by formulas.
    This means that the same natural deduction rule can be applied in \(\sf{HA}k\), since natural deduction rules are stable under substitution of relation symbols.
    \item[Transitive well-founded relation.]
    We can show in \(\sf{HA}k\) that \(<\) defined by \((a<b)\defeq\exists_{x:\bb N}\,(\sf{suc}(a+x)=b)\) is a well-order, which means in particular that we can derive the transitivity rule and the well-founded induction principle of \((\sf{CHA}k)_\ind\).
\end{description}

\section{Omitted proofs}\label{app: proofs}

\subsection{Interderivability of \texorpdfstring{ind$_x$}{ind\_x} and \texorpdfstring{ind$_x^{\sf{conc}}$}{ind\_xconc}}\label{app: proof interderivability}
    We show that the following two rules are interderivable given the other rules of \(\cal C_\ind\):
    \begin{gather*}
        \prfstack[r]{ind\(_x\),}
            {\cal X,x:S;\Gamma,\phi_{\forall_{\cal X}\,(\Gamma\to\delta)}\vdash\delta}
            {\cal X,x:S;\Gamma\vdash\delta} \qquad
        \prfstack[r]{ind\(_x^{\sf{conc}}\),}
            {\cal X,x:S;\Gamma,\phi_\delta\vdash\delta}
            {\cal X;\Gamma\vdash\forall_{x:S}\,\delta}
    \end{gather*} where \(\phi_\delta\) and \(\phi_{\forall_{\cal X}\,(\Gamma\to\delta)}\) abbreviate the induction hypotheses \begin{align*}
        \phi_\delta&\coloneqq\forall_{x'}\,(x>x'\to\delta[x'/x]), \\
        \phi_{\forall_{\cal X}(\Gamma\to\delta)}&\coloneqq\forall_{x'}\,(x>x'\to\forall_{\cal X}\,(\Gamma\to\delta)[x'/x]).
    \end{align*}
    For notational convenience, we work out the case where \(\cal X\) and \(\Gamma\) only contain a single variable assignment \(z:T\) and formula \(\gamma\): \begin{gather*}
        \prftree[r]{\(\to\)-elim,}
            {\prftree[r]{\(\forall\)-elim}
                {\prftree[r]{weaken}
                    {\prftree[r]{ind\(_x^{\sf{conc}}\)}
                        {\prftree[r]{\(\forall\)-intro}
                            {\prftree[r]{\(\to\)-intro}
                                {z:T,x:S;\gamma,\phi_{\forall_ z(\gamma\to\delta)}\vdash\delta}
                                {z:T,x:S;\phi_{\forall_ z(\gamma\to\delta)}\vdash\gamma\to\delta}}
                            {x:S;\phi_{\forall_ z(\gamma\to\delta)}\vdash\forall_z\,(\gamma\to\delta)}}
                        {\vdash\forall_x\,\forall_z\,(\gamma\to\delta)}}
                    {z:T,x:S;\gamma\vdash\forall_x\,\forall_z\,(\gamma\to\delta)}}
                {z:T,x:S;\gamma\vdash\gamma\to\delta}}
            {\prftree[r]{id}
                {z:T,x:S;\gamma\vdash\gamma}}
            {z:T,x:S;\gamma\vdash\delta} \\[2ex]
        \prftree[r]{\(\forall\)-intro.}
            {\prftree[r]{ind\(_x\)}
                    {\prftree[r]{\(\to\)-elim}
                        {\prftree[r]{\(\to\)-intro}
                            {\prftree[r]{weaken}
                                {\gamma,\phi_\delta\vdash\delta}
                                {\gamma,\phi_{\forall_z(\gamma\to\delta)},\phi_\delta\vdash\delta}}
                            {\gamma,\phi_{\forall_z(\gamma\to\delta)}\vdash\phi_\delta\to\delta}}
                        {\prftree
                            {\prftree{\cdots}}
                            {\gamma,\phi_{\forall_z(\gamma\to\delta)}\vdash\phi_\delta}}
                        {\gamma,\phi_{\forall_z(\gamma\to\delta)}\vdash\delta}}
                {\gamma\vdash\delta}}
            {\gamma\vdash\forall_x\,\delta}
    \end{gather*}
    Here the dots stand for the following proof, which uses the fact that \(x\) cannot appear free in \(\gamma\) for the ind\(_x^{\sf {conc}}\)-rule, and omits the right branches of the \(\to\)-elim rules (since they are easy to fill in): \begin{gather*}
        \prftree[r]{\(\forall\)-intro}
            {\prftree[r]{\(\to\)-intro}
                {\prftree[r]{\(\to\)-elim}
                    {\prftree[r]{\(\forall\)-elim}
                        {\prftree[r]{\(\to\)-elim}
                            {\prftree[r]{\(\forall\)-elim}
                                {\prftree[r]{id}
                                    {\gamma,\phi_{\forall_z(\gamma\to\delta)},x>x'\vdash\phi_{\forall_z(\gamma\to\delta)}}\hspace*{2cm}}
                                {\gamma,\phi_{\forall_z(\gamma\to\delta)},x>x'\vdash x>x'\to\forall_z\,(\gamma\to\delta[x'/x])}}
                            {\gamma,\phi_{\forall_z(\gamma\to\delta)},x>x'\vdash\forall_z\,(\gamma\to\delta[x'/x])}}
                        {\gamma,\phi_{\forall_z(\gamma\to\delta)},x>x'\vdash\gamma\to\delta[x'/x]}}
                    {\gamma,\phi_{\forall_z(\gamma\to\delta)},x>x'\vdash\delta[x'/x]}}
                {\gamma,\phi_{\forall_z(\gamma\to\delta)}\vdash x>x'\to\delta[x'/x]}}
            {\gamma,\phi_{\forall_z(\gamma\to\delta)}\vdash\phi_\delta}
    \end{gather*}
    By iterating some of the steps, this can easily be shown for the general case.

\subsection{Proof of \texorpdfstring{\autoref{lem: existence preserved progress name}}{Lemma 2}}\label{app: proof of lem exist progress name}
    Since $(n_i)_i$ is a branch of a proof, we know that it has a progressing trace $(j_i)_{i\ge k}$.
    Consider the corresponding sequence of stacks $(\sf{Stack}(n_i)_{j_i})_{i\ge k}$.
    Let $m>0$ be the maximal value such that from some $k'\geq k$ onwards, $|\sf{Stack}(n_i)_{j_i}|\geq m$ for every $i\ge k'$; such an $m$ must exist because the stacks are nonempty by construction and, by regularity and \autoref{lem: finitely many names}, their size is bounded.
    
    We show that the first \(m\) names in the sequence of stacks $(\sf{Stack}(n_i)_{j_i})_{i\ge k'}$ are eventually fixed.
    These names cannot change solely because of a reset, as the size of the stack would drop below $m$, and since $(j_i)_{i\ge k'}$ is a trace, we cannot start a new stack along this sequence either.
    Thus these $m$ names can only change when picking an older candidate stack than the one provided by our sequence.
    However, such a change will only change the first \(m\) names if the older candidate contains a name that is older than the previous \(m\)-th name, and this can only happen finitely often since there are only finitely many names older than the name that initially is at the \(m\)-th place.
    We may therefore take $k''\ge k'$ large enough that the first $m$ names in $(\sf{Stack}(n_i)_{j_i})_{i\ge k''}$ are fixed, where we call the \(m\)-th name \(a\).
    
    As the trace is infinitely progressing, following our trace will infinitely often provide a candidate that contains more than \(m\) names.
    This candidate need not be chosen, but if it is not, then the candidate that was chosen must have an older name in the symmetric difference, and therefore also have more than \(m\) names, since the first \(m\) names are the same.
    As we picked $m$ as the maximal recurring height of the stacks, there must be infinitely many $a$-resets along $(\sf{Stack}(n_i)_{j_i})_{i\ge k''}$.

\subsection{Proof of \texorpdfstring{\autoref{prop:SD transitive}}{Proposition 2}}\label{app: proof of prop SD trans}
We first define the notion of \emph{reachability} and prove two lemmas. \autoref{prop:SD transitive} will then follow directly from the second lemma.
\begin{definition}[reachable]
    Let $\pi_{\sf{cyc}}$ be a cyclic representation of some $\cal{C}$-proof.
    For any $n,n'\in\sf{RepNode}$, we say that $n'$ is \emph{reachable from $n$} (possibly using back-edges) if there exist nodes $n_0,\dots,n_{k-1}\in \sf{RepNode}$ such that $n_0=n$, $n_{k-1}=n'$ and for every $i<k-1$, either $n_{i+1}$ is a child of $n_i$, or $n_i$ is a bud with $\sf{sprout}(n_i)=n_{i+1}$.
\end{definition}

\begin{lemma}\label{lem: lower bound for reachable buds}
    Any finite set of mutually reachable buds contains a most important bud.
\end{lemma}
\begin{proof}
    We only provide a sketch, and refer the reader to \cite[Proposition 13]{leigh_gtc_2024} for details.

    Let $B$ be a set of mutually reachable buds. Then there exists a sequence $\overline{n}\coloneqq n_0,\dots,n_{k}$ in $\sf{RepNode}$ that contains every bud in $B$, such that $n_0=n_{k}$ and for every $i<k$, either $n_{i+1}$ is a child of $n_i$, or $n_i$ is a bud with $\sf{sprout}(n_i)=n_{i+1}$. Since $\overline{n}$ describes a `cycle with back-edges' through $\sf{RepNode}$, note that $\overline{n}$ must contain a lowest node $n_l$, and that this node must be the sprout of some $b_l\in B$. We can then define a well-founded order $\preceq$ on $B$ based on the minimal number of back-edges that a bud needs to reach the bud $b_l$. 

    We can then prove by induction on $\preceq$ that for every $b\in B$, the downset $\sf{Down}_\preceq(b)$ contains a most important bud. The key observation for the inductive step is that, if $b'\prec b$ such that there is no $b'\prec b''\prec b$, then $\sf{sprout}(b)$ lies on the path from $\sf{sprout}(b')$ to $b'$; the latter implies that $b$ and the most important bud in $\sf{Down}_\preceq(b')$ are comparable in terms of importance. 
\end{proof}

\begin{lemma}\label{lem:SD simpel}
  Every $\cal{C}$-proof $\pi$ has a cyclic representation \(\pi_{\sf{cyc}}\) with progress map $\sf{prog}$ that satisfies the properties of \autoref{thm:existence cyclic rep with prog name} and in addition: if $b'\triangleright b$, then $b'$ is more important than $b$.
\end{lemma}
\begin{proof}
    Let the cyclic representation \(\pi_{\sf{cyc}}\) with progress map $\sf{prog}$ be given by \autoref{thm:existence cyclic rep with prog name}. The main idea is that, to obtain the desired representation $\pi'_{\sf{cyc}}$, we will unfold \(\pi_{\sf{cyc}}\) while only introducing a back-edge for a bud if all more important buds already have an `available' sprout node lower in the tree. 
    
    To make this precise, we first define a map $\sf{cycRep}:\sf{Node}\to\sf{RepNode}$ that sends each node of $\pi$ to a representative in $\sf{RepNode}$. We define $\sf{cycRep}$
    by induction on the height of the node: for a branch \((n_i)_i\), let \((j_i)_i\) be the indices such that \(n_{i+1}=\sf{Child}(n_i)_{j_i}\), and: \begin{align*}
        \sf{cycRep}(n_0)&\coloneqq n_0, \\
        \sf{cycRep}(n_{i+1})&\coloneqq\begin{cases}
            \sf{Child}(\sf{sprout}(\sf{cycRep}(n_i)))_{j_i}\quad\text{if this is defined,} \\
            \phantom{\sf{sprout}(}\sf{Child}(\sf{cycRep}(n_i))_{j_i}\quad\text{otherwise.}
        \end{cases}
    \end{align*}
    Note that $\sf{cycRep}$ sends a bud (and `copies' thereof) to itself, and not to its sprout.

    We now define partial functions \(
        \sf{avail}(n_i):\sf{Bud}\rightharpoonup\sf{Node}
    \), which indicate which (if any) potential sprout is available for a given bud, using the following algorithm (with reassignments): \begin{itemize}
        \item if \(i=0\), initiate \(\sf{avail}(n_0)\) as the empty function, otherwise \(\sf{avail}(n_{i})\coloneqq\sf{avail}(n_{i-1})\);
        \item then, if \({\sf{cycRep}(n_{i})}=b'\) for some bud \(b'\notin\sf{dom}(\sf{avail}(n_{i}))\), set \(\sf{avail}(n_{i})(b')\coloneqq n_{i}\) and remove all buds from the domain of \(\sf{avail}(n_{i})\) that are not more important than \(b'\) (including incomparable ones).
    \end{itemize}

    We build a new cyclic representation $\pi'_{\sf{cyc}}$ with buds $\sf{Bud}'$, sprout map \(\sf{sprout}'\) and progress map $\sf{prog}'$.
    For every infinite branch $(n_i)_{i<\omega}$, we let \(t\) be the least index such that $\sf{avail}(n_t)(\sf{cycRep}(n_t))$ is defined as some \(n_s\) with \(s<t\), in which case we let $n_t$ be in $\sf{Bud}'$, \(\sf{sprout}'(n_t)\coloneqq n_s\), \(\sf{prog}'(n_t)\coloneqq\sf{prog}(\sf{cycRep}(n_t))\).
    This gives a finite tree by K\"onig's lemma.

    To see the existence of such an index, consider the infinite sequence $( \sf{cycRep}(n_{i}))_{i<\omega}$ and let $k$ be such that for every $i\ge k$, the node $\sf{cycRep}(n_{i})$ occurs infinitely often in this sequence. Note that all the buds in $( \sf{cycRep}(n_{i}))_{i\ge k}$ can reach each other, so by \autoref{lem: lower bound for reachable buds} there exists a $b\in\sf{Bud}$ that is more important than any bud in $( \sf{cycRep}(n_{i}))_{i\ge k}$.
    Now let $s\geq k$ be such that $\sf{cycRep}(n_{s})=b$.
    If \(b\notin\sf{dom}(\sf{avail}(n_{s}))\), then a value for \(b\) is added to \(\sf{avail}(n_s)\), otherwise a value is already there. 
    Moreover, this value can never be removed, as $b$ is more important than any bud $\sf{cycRep}(n_{i})$ for $i\geq s$.
    So, we can take some $t>s$ such that $\sf{cycRep}(n_t)=b$ and \(\sf{avail}(n_t)(\sf{cycRep}(n_t))\) is defined and strictly below $n_t$.
    
    We check that $\pi'_{\sf{cyc}}$ has the desired properties. With a straightforward induction, one shows that \(\sf{cycRep}\) preserves \(\sf{Name}\) and \(\sf{Stack}\) because the map \(\sf{sprout}\) does.
    Thus \(\pi_{\sf{cyc}}'\) inherits the properties of \autoref{thm:existence cyclic rep with prog name} from  \(\pi_{\sf{cyc}}\).
    Now, suppose that we have two buds \(b,b'\in\sf{Bud}'\) with $b'\triangleright b$, meaning that \(\sf{sprout}'(b')\) lies strictly below $\sf{sprout}'(b)$ and $\sf{sprout}'(b)$ lies below \(b'\). By definition of $\pi'_{\sf{cyc}}$, we then have $\sf{avail}(b)(\sf{cycRep}(b))=\sf{sprout}'(b)$, so we initially added a value for $\sf{cycRep}(b)$ to $\sf{avail}(\sf{sprout}'(b))$. If $\sf{cycRep}(b')$ is not more important than $\sf{cycRep}(b)$, then we would have removed a possible value for $\sf{cycRep}(b')$ from $\sf{avail}(\sf{sprout}'(b))$. However, as \(\sf{sprout}'(b')\) lies strictly below $\sf{sprout}'(b)$ and $\sf{sprout}'(b)$ lies below \(b'\), this would mean that $\sf{avail}(b')$ cannot assign the value $\sf{sprout}'(b')$ to $\sf{cycRep}(b')$, contradicting that $\sf{sprout}'(b')$ is the sprout of $b'$. Thus $\sf{cycRep}(b')$ is more important than $\sf{cycRep}(b)$, and moreover, as $\sf{sprout}'(b)$ lies between $\sf{sprout}'(b')$ and $b'$, every name $a$ in $\sf{Name}(b')\upharpoonright \sf{prog}'(b')$ must satisfy $\sf{var}(b,a)=\sf{var}(b',a)$ (see the comment below \autoref{thm:existence cyclic rep with prog name}), which implies that $b'$ is more important than $b$.
\end{proof}

\begin{proof}[of \autoref{prop:SD transitive}]
    This follows directly from \autoref{lem:SD simpel} and the fact that the importance order is transitive. 
\end{proof}

\subsection{Proof that \texorpdfstring{$\rho_b$}{rho\_b} satisfies the requirements of \texorpdfstring{\autoref{def: well org cyc rep}}{Definition 6}}\label{app: proof that the renaming works}
Let $b$ be a bud in $\pi'_{\sf{cyc}}$ with sprout $s$, and let $s'$ be the unique child of $s$. Let $\sf{cov}(b)$ be a name that uniformly covered $\sf{prog}(b)$ before performing the reset on $\sf{prog}(b)$ at $b$. Recall that $x_{\sf{prog}_b}\defeq\sf{var}(b,\sf{prog}(b))$ and that $\rho_b$ is defined by:
        \begin{align*}
        x[\rho_b]\coloneqq \begin{cases}
            x_{b,j} & \text{if }x=x_{s,j}\text{ for some trace object }j, \\
           \sf{var}(b,\sf{cov}(b)) &\text{if }x=x_{\sf{prog}_b}, \\
            \sf{var}(b,a_x) & \text{otherwise},
        \end{cases}
        \end{align*}
where $a_x\in \sf{Name}(s)=\sf{Name}(b)$ such that $x=\sf{var}(s,a_x)$. We show that the substitution $\rho_b$ satisfies the requirements of \autoref{def: well org cyc rep}. 

We first check that $\rho_b$ indeed maps variables from $\cal X_{s'}$ to variables in $\cal X_{b}$ of the same sort. By construction of the variable contexts, it is clear that $\cal X_{b}$ contains the variables $x_{b,j}$ and the variables $\sf{var}(b,a_x)$ that satisfy $a_x\in\sf{Name}(b)$. That $\cal X_{b}$ also contains the variable $\sf{var}(b,\sf{cov}(b))$ follows from the fact that the name $\sf{cov}(b)$ was present in $\sf{Name}(b)$ before the reset on $\sf{prog}(b)$ at $b$ and the fact that we only forget variables at sprouts (and $b$ cannot also be a sprout). That $\rho_b$ respects sorts follows from the fact that $b$ and $s$ have the same label and the fact that names on the same stack can only refer to variables of the same sort (as the edges in a size-change graph only connect objects of the same sort). 

To see that $\rho_b$ is well-defined, note that the first two cases in the definition are exclusive because (by the requirements of \autoref{thm:existence cyclic rep with prog name}) the name $\sf{prog}(b)$ was reset at $s$ and thus $x_{\sf{prog}_b}$ cannot be one of the $x_{s,j}$ (as those cannot carry a covered name at $s$).

Now note that, for any node $n$ in $\pi'_{\sf{cyc}}$, the inequalities in $\cal Q_n$ are all `true', in the sense that $(y\geq z)\in\cal Q_n$ only if $y$ is an ancestor of $z$ and $(y> z)\in\cal Q_n$ only if $y$ is a strict ancestor of $z$; this follows from the fact that we only add inequalities to $\cal Q_n$ that are induced by (compositions of) edges in the size-change graphs. 
It is then not hard to see that the set $\cal Q_{s'}$ contains \textbf{exactly} the inequalities of the following form:
\begin{itemize}
    \item[i$_{s'}$.] $\sf{var}(s,a)\geq x_{s,j}$, where $a$ is the topmost name in $\sf{Stack}(s)_j$;
    \item[ii$_{s'}$.] $\sf{var}(s,a)> \sf{var}(s,a')$, where $a,a'\in\sf{Stack}(s)_j$ for some $j$ and $a$ is strictly older than $a'$.
\end{itemize} 
Similarly, the set $\cal Q_b$ contains \textbf{at least} the inequalities of the following form: 
\begin{itemize}
    \item[i$_{b}$.] $\sf{var}(b,a)\geq x_{b,j}$, where $a\in \sf{Stack}(b)_j$ or $a=\sf{cov}(b)$;
    \item[ii$_{b}$.] $\sf{var}(b,a)\geq \sf{var}(b,a')$, where $a,a'\in\sf{Stack}(b)_j\cup\{\sf{cov}(b)\}$ for some $j$ and $a$ is  older than $a'$;
    \item[iii$_{b}$.] $\sf{var}(b,a)> \sf{var}(b,a')$, where $a,a'\in\sf{Stack}(b)_j\cup\{\sf{cov}(b)\}$ for some $j$ and $a$ is strictly older than $a'$.
\end{itemize} 
To see this, note that we close the inequalities under composition at $\cal C$-rules, and that any sprout and bud must be separated by a $\cal C$-rule. Moreover, even though $\sf{cov}(b)\notin \sf{Name}(b)$ due to the reset on $\sf{prog}(b)$ at $b$, we still have the inequalities involving $\sf{var}(b,\sf{cov}(b))$ in $\cal Q_b$ because we only remove inequalities for nameless variables at sprout nodes (and $b$ cannot also be a sprout). 

Using these observations, we now show that $\rho_b$ satisfies each of the four properties listed in \autoref{def: well org cyc rep}:
\begin{enumerate}
    \item We have \((x_{\sf{prog}_{b\phantom'}}>\sf{var}(b,\sf{cov}(b)))\in\cal Q_b\) by clause iii$_b$, because $\sf{cov}(b)$ covered $\sf{prog}(b)$ (and thus $\sf{prog}(b)$ is strictly older) in at least one of the stacks at $b$.
    \item Let $b'\triangleright^*b$. Then $b'$ is more important than $b$, so $\sf{prog}(b')$ is older than  $\sf{prog}(b)$ at $b$ and  $x_{\sf{prog}_{b'}}\defeq\sf{var}(b',\sf{prog}(b'))=\sf{var}(b,\sf{prog}(b'))$. We show that \((x_{\sf{prog}_{b'}}\ge x_{\sf{prog}_{b'}}[\rho_b])\in\cal Q_b\) by distinguishing three cases (following the definition of $\rho_b$):
    \begin{itemize}
        \item If $x_{\sf{prog}_{b'}}=x_{s,j}$ for some $j$, then $\sf{prog}(b')$ occurs in $\sf{Stack}(s)_j=\sf{Stack}(b)_j$ and thus \((x_{\sf{prog}_{b'}}\ge x_{b,j})\in\cal Q_b\) by clause i$_b$.
        \item If $x_{\sf{prog}_{b'}}=x_{\sf{prog}_{b}}$, then $x_{\sf{prog}_{b'}}[\rho_b]=\sf{var}(b,\sf{cov}(b))$ and so the statement follows from clause ii$_b$ (because $\sf{prog}(b')$ is older than $\sf{cov}(b)$).
        \item Otherwise, we have $x_{\sf{prog}_{b'}}[\rho_b] = \sf{var}(b,\sf{prog}(b'))=x_{\sf{prog}_{b'}}$, so the statement follows since \((x_{\sf{prog}_{b'}}\ge x_{\sf{prog}_{b'}})\in\cal Q_b\) by clause ii$_b$.
    \end{itemize}
    \item  We distinguish two cases (following the clauses i$_{s'}$ and ii$_{s'}$):
        \begin{itemize}
            \item Consider $(\sf{var}(s,a)\geq x_{s,j})\in\cal Q_{s'}$ with $a$ topmost in $\sf{Stack}(s)_j$.  If $a=\sf{prog}(b)$ then $\sf{var}(s,a)[\rho_b]=x_{\sf{prog}_b}[\rho_b]=\sf{var}(b,\sf{cov}(b))$, and if not then $\sf{var}(s,a)[\rho_b]=\sf{var}(b,a)$. So in each case we have $(\sf{var}(s,a)[\rho_b]\geq x_{b,j})\in\cal Q_{b}$ by clause i$_b$.
            \item Consider $(\sf{var}(s,a)> \sf{var}(s,a'))\in\cal Q_{s'}$ with $a,a'\in \sf{Stack}(s)_j$ for some $j$ and $a$ strictly older than $a'$. Then $a\neq \sf{prog}(b)$, so $\sf{var}(s,a)[\rho_b]=\sf{var}(b,a)$. If $a'=\sf{prog}(b)$, then $a$ is older than $\sf{cov}(b)$ and $\sf{var}(s,a')[\rho_b]=\sf{var}(b,\sf{cov}(b))$, and if $a'\neq \sf{prog}(b)$ then $\sf{var}(s,a')[\rho_b]=\sf{var}(b,a')$. So in each case we have $(\sf{var}(s,a)[\rho_b]>\sf{var}(s,a')[\rho_b])\in\cal Q_{b}$ by clause iii$_b$.
        \end{itemize}
    \item We have $x_{b,j}=x_{s,j}[\rho_b]$ by definition.
\end{enumerate}

\section{Examples of reset proofs}\label{app:reset}

Recall that in \autoref{exam:run} we have given functions with the following size-change graphs: \begin{gather*}
    \begin{tikzpicture}[yscale=0.8,remember picture]
        \node at (0,0) {\(\mathsf{inter}(\subnode{x0C}{x_0},\subnode{x1C}{x_1})\)};
        \node at (0,1) {\(\mathsf{inter}(\subnode{y0C}{y_0},\subnode{y1C}{y_1})\)};
        \draw[->,green] (x0C.north) -- node[pos=.3,left=7pt]  {\(\rotatebox{90}{\(>\)}\)}   (y1C.south);
        \draw[->]       (x1C.north) -- node[pos=.3,right=7pt] {\(\rotatebox{90}{\(\ge\)}\)} (y0C.south);
    \end{tikzpicture}~~~
    \begin{tikzpicture}[yscale=0.8,remember picture]
        \node at (0,0) {\(\mathsf{Ack}(\subnode{x0A}{x_0},\subnode{x1A}{x_1})\)};
        \node at (0,1) {\(\mathsf{Ack}(\subnode{y0A}{y_0},\subnode{y1A}{y_1})\)};
        \draw[->,green] (x0A.north) -- node[left] {\(\rotatebox{90}{\(>\)}\)} (y0A.south);
    \end{tikzpicture}~~~
    \begin{tikzpicture}[yscale=0.8,remember picture]
        \node at (0,0) {\(\mathsf{Ack}(\subnode{x0A'}{x_0},\subnode{x1A'}{x_1})\)};
        \node at (0,1) {\(\mathsf{Ack}(\subnode{y0A'}{y_0},\subnode{y1A'}{y_1})\)};
        \draw[->,green] (x0A'.north) -- node[left] {\(\rotatebox{90}{\(>\)}\)} (y0A'.south);
    \end{tikzpicture}~~~
    \begin{tikzpicture}[yscale=0.8,remember picture]
        \node at (0,0) {\(\mathsf{Ack}(\subnode{x0B}{x_0},\subnode{x1B}{x_1})\)};
        \node at (0,1) {\(\mathsf{Ack}(\subnode{y0B}{y_0},\subnode{y1B}{y_1})\)};
        \draw[->]       (x0B.north) -- node[left]  {\(\rotatebox{90}{\(\ge\)}\)} (y0B.south);
        \draw[->,green] (x1B.north) -- node[right] {\(\rotatebox{90}{\(>\)}\)}   (y1B.south);
    \end{tikzpicture}~~~
    \begin{tikzpicture}[yscale=0.8,remember picture]
        \node at (0,0) {\(f(\subnode{x0D}{x_0},\subnode{x1D}{x_1})\)};
        \node at (0,1) {\(g(\subnode{y0D}{y_0})\)};
        \draw[->,in=90] (x0D.north) to[out=60,in=-90] node[left=7pt] {\(\rotatebox{90}{\(\ge\)}\)} (y0D.south);
        \draw[->,in=-90] (x1D.north) to[out=120,in=-90] node[right=7pt] {\(\rotatebox{90}{\(\ge\)}\)} (y0D.south);
    \end{tikzpicture}~~~
    \begin{tikzpicture}[yscale=0.8,remember picture]
        \node at (0,0) {\(g(\subnode{x0E}{x_0})\)};
        \node at (0,1) {\(f(\subnode{y0E}{y_0},\subnode{y1E}{y_1})\)};
        \draw[->,green] (x0E.north) -- node[left] {\(\rotatebox{90}{\(>\)}\)} (y0E.south);
    \end{tikzpicture}
\end{gather*}
In \autoref{exam:run1} we have seen a proof system where each of these functions has a unique infinite derivation showing that the function terminates: the tree whose branches are given by sequences of recursive calls starting from the function.
For example, \(\sf{inter}\) calls itself once, so we have a unique infinite derivation, with one infinite branch: \[
    \prftree[r]{\(\sf{inter}\)}
        {\prftree[r]{\(\sf{inter}\)}
            {\prftree[r]{\(\sf{inter}\)}
                {\vdots}
                {\sf{inter}{\downarrow}}}
            {\sf{inter}{\downarrow}}}
        {\sf{inter}{\downarrow}}
\]
Since the function satisfies the size-change termination condition, this derivation is a proof, and it becomes a lot more insightful if we add distinct variable names for the trace objects and visualize the size-change graphs: \[
\begin{tikzpicture}[yscale=0.8,remember picture, baseline={([yshift=-.5ex]current bounding box.center)}]
            \node (xkq) at (0,0) {\(\mathsf{inter}(\subnode{x0kq}{x_0},\subnode{x1kq}{x_1}){\downarrow}\)};
            \node (ykq) at (0,1) {\(\mathsf{inter}(\subnode{y0kq}{y_0},\subnode{y1kq}{y_1}){\downarrow}\)};
            \node (zkq) at (0,2) {\(\mathsf{inter}(\subnode{z0kq}{z_0},\subnode{z1kq}{z_1}){\downarrow}\)};
            \node (dkq) at (0,2.75) {\(\vdots\)};
            \draw[->,green] (x0kq.north) to node[pos=.3,left=7pt]  {\(\rotatebox{90}{\(>\)}\)}   (y1kq.south);
            \draw[->]       (x1kq.north) to node[pos=.3,right=7pt] {\(\rotatebox{90}{\(\ge\)}\)} (y0kq.south);
            \draw[->,green] (y0kq.north) to node[pos=.3,left=7pt]  {\(\rotatebox{90}{\(>\)}\)}   (z1kq.south);
            \draw[->]       (y1kq.north) to node[pos=.3,right=7pt] {\(\rotatebox{90}{\(\ge\)}\)} (z0kq.south);
        \end{tikzpicture}
\]
We can apply \autoref{def:safra} to get annotations to every node of this derivation.
We can visualize this as follows: \[
    \begin{tikzpicture}[remember picture, xscale=1.5]
        \node (N0) at (0,0) {\(\mathsf{inter}(\subnode{N00}{a},\subnode{N01}{b}){\downarrow}\)};
        \node (N1) at (0,1) {\(\mathsf{inter}(\subnode{N10}{b},\subnode{N11}{a\hcancel c}){\downarrow}\)};
        \node (N2) at (0,2) {\(\mathsf{inter}(\subnode{N20}{a},\subnode{N21}{b\hcancel c}){\downarrow}\)};
        \node (Nkq) at (0,2.75) {\(\vdots\)};
        \path[->, out=75, in=-105, green] (N00.north) edge node[very near start,left] {\(\rotatebox{90}{\(>\)}\)} (N11.south);
        \path[->, out=105, in=-75] (N01.north) edge node[very near start,right] {\(\rotatebox{90}{\(\ge\)}\)} (N10.south);
        \path[->, out=75, in=-105, green] (N10.north) edge node[very near start,left] {\(\rotatebox{90}{\(>\)}\)} (N21.south);
        \path[->, out=105, in=-75] (N11.north) edge node[very near start,right] {\(\rotatebox{90}{\(\ge\)}\)} (N20.south);
      \end{tikzpicture}
\]
Here we have written the stack for every trace object, where we have \smash{\hcancel{\text{crossed out}}} names that are removed by a reset.
Note that \(a\) denotes \(x_0\) and \(b\) denotes \(x_1\) throughout, while \(c\) is a name for \(y_1\) that is removed by a reset (since it uniformly covers \(a\)) and later reused for \(z_1\) before immediately being removed again (since it uniformly covers \(b\)).
We also see that the annotations will be periodic with period 2, since the annotations of a child only depend on the annotations of the parent and the size-change graph from the parent to the child.
The stacks can be seen as relevant history for each variable: a list of names for some of its ancestors.
In this case, the only relevant history is the oldest ancestor.

We now focus on showing some proofs that satisfy the conclusion of \autoref{thm:existence cyclic rep with prog name}, which give the intuition for why this theorem holds: we start by annotating the tree from the root up, and stop when we start seeing repeats in annotations.
In particular, for \(\sf{inter}\), we obtain the following reset proof in this way: \[
\begin{tikzpicture}[remember picture, xscale=1.5]
    \node (n-0) at (0,0) {\(\mathsf{inter}(\subnode{n-00}{a},\subnode{n-01}{b}){\downarrow}\)};
    \node (n-1) at (0,1) {\(\mathsf{inter}(\subnode{n-10}{b},\subnode{n-11}{a\hcancel c}){\downarrow}\)};
    \node (n-2) at (0,2) {\(\mathsf{inter}(\subnode{n-20}{a},\subnode{n-21}{b\hcancel c}){\downarrow}\)};
    \node (n-3) at (0,3) {\(\mathsf{inter}(\subnode{n-30}{b},\subnode{n-31}{a\hcancel c}){\downarrow}\)};
    \node (n-4) at (0,4) {\(\mathsf{inter}(\subnode{n-40}{a},\subnode{n-41}{b\hcancel c}){\downarrow}\)};
    \path[->, out=75, in=-105, green] (n-00.north) edge node[very near start,left] {\(\rotatebox{90}{\(>\)}\)} (n-11.south);
    \path[->, out=105, in=-75] (n-01.north) edge node[very near start,right] {\(\rotatebox{90}{\(\ge\)}\)} (n-10.south);
    \path[->, out=75, in=-105, green] (n-10.north) edge node[very near start,left] {\(\rotatebox{90}{\(>\)}\)} (n-21.south);
    \path[->, out=105, in=-75] (n-11.north) edge node[very near start,right] {\(\rotatebox{90}{\(\ge\)}\)} (n-20.south);
    \path[->, out=75, in=-105, green] (n-20.north) edge node[very near start,left] {\(\rotatebox{90}{\(>\)}\)} (n-31.south);
    \path[->, out=105, in=-75] (n-21.north) edge node[very near start,right] {\(\rotatebox{90}{\(\ge\)}\)} (n-30.south);
    \path[->, out=75, in=-105, green] (n-30.north) edge node[very near start,left] {\(\rotatebox{90}{\(>\)}\)} (n-41.south);
    \path[->, out=105, in=-75] (n-31.north) edge node[very near start,right] {\(\rotatebox{90}{\(\ge\)}\)} (n-40.south);
    \path[->, out=180, in=180, blue, dashed] (n-4.west) edge node[left] {\(b\)} (n-2.west);
  \end{tikzpicture}
\]

For the remaining examples we will work with a slightly relaxed version of reset proof.
In \autoref{thm:existence cyclic rep with prog name}, we have chosen to require that both the bud and the sprout have executed a reset on the progressing variable.
In the example above, this means that we have to do the periodic pattern twice before we can introduce a back-edge.
It would also have been fine to require only that the bud has just executed a reset on the progressing variable, without the requirement that this is also the case for the sprout.
This is a trade-off: \begin{itemize}
    \item Requiring a reset at both the bud and the sprout is easier for the proofs in \autoref{sec:safra} and \autoref{app: proofs}, since we know that the progressing name is never one of the new names of the sprout, which means that we have less cases to cover.
    \item Only requiring a reset at the bud produces smaller examples, which is why we make this choice here.
\end{itemize}
For example, with this weaker condition, the following smaller proof already suffices: \[
\begin{tikzpicture}[remember picture, xscale=1.5]
    \node (n0) at (0,0) {\(\mathsf{inter}(\subnode{n00}{a},\subnode{n01}{b}){\downarrow}\)};
    \node (n1) at (0,1) {\(\mathsf{inter}(\subnode{n10}{b},\subnode{n11}{a\hcancel c}){\downarrow}\)};
    \node (n2) at (0,2) {\(\mathsf{inter}(\subnode{n20}{a},\subnode{n21}{b\hcancel c}){\downarrow}\)};
    \path[->, out=75, in=-105, green] (n00.north) edge node[very near start,left] {\(\rotatebox{90}{\(>\)}\)} (n11.south);
    \path[->, out=105, in=-75] (n01.north) edge node[very near start,right] {\(\rotatebox{90}{\(\ge\)}\)} (n10.south);
    \path[->, out=75, in=-105, green] (n10.north) edge node[very near start,left] {\(\rotatebox{90}{\(>\)}\)} (n21.south);
    \path[->, out=105, in=-75] (n11.north) edge node[very near start,right] {\(\rotatebox{90}{\(\ge\)}\)} (n20.south);
    \path[->, out=180, in=180, blue, dashed] (n2.west) edge node[left] {\(b\)} (n0.west);
  \end{tikzpicture}
\]
Note that this is precisely the well-organised cyclic representation of \autoref{exam:importprogres}, and that the previous reset proof is an unfolding of this.
This last point holds generally: if we unfold each cycle at least once, then a proof that satisfies the relaxed condition can be turned into one that satisfies the strict condition.

For \(\sf{Ack}\), this relaxed procedure produces the following: \[
    \begin{tikzpicture}[remember picture, xscale=2.8]
    \node (o0)  at ( 0,0) {\(\mathsf{Ack}(\subnode{o00}{a},\subnode{o01}{b}){\downarrow}\)};
    \node (o1)  at (-1,1) {\(\mathsf{Ack}(\subnode{o10}{a\hcancel{c}},\subnode{o11}{d}){\downarrow}\)};
    \node (o1') at ( 0,1) {\(\mathsf{Ack}(\subnode{o10'}{a},\subnode{o11'}{b\hcancel{d}}){\downarrow}\)};
    \node (o2)  at (-2,2) {\(\mathsf{Ack}(\subnode{o20}{a\hcancel{c}},\subnode{o21}{b}){\downarrow}\)};
    \node (o2') at (-1,2) {\(\mathsf{Ack}(\subnode{o20'}{a},\subnode{o21'}{d\hcancel{b}}){\downarrow}\)};
    \path[->, out=180, in=-60, green] (o00.north) edge node[sloped,below,near end] {\(<\)} (o10.south);
    \path[->] (o00.north) edge node[left] {\(\rotatebox{90}{\(\ge\)}\)} (o10'.south);
    \path[->, green] (o01.north) edge node[right] {\(\rotatebox{90}{\(>\)}\)} (o11'.south);
    \path[->, out=180, in=-60, green] (o10.north) edge node[sloped,below,near end] {\(<\)} (o20.south);
    \path[->] (o10.north) edge node[left] {\(\rotatebox{90}{\(\ge\)}\)} (o20'.south);
    \path[->, green] (o11.north) edge node[right] {\(\rotatebox{90}{\(>\)}\)} (o21'.south);
    \path[->, out=180, in=180, in looseness=2, blue, dashed] (o2.west) edge node[below left] {\(a\)} (o0.west);
    \path[->, out=-30, in=30, red, dashed] (o1'.east) edge node[above right] {\(b\)} (o0.east);
    \path[->, out=-30, in=30, cyan, dashed] (o2'.east) edge node[above right] {\(d\)} (o1.east);
  \end{tikzpicture}
\] where we have only drawn each distinct size-change graph once.
This is still bigger than the well-organized cyclic representation we saw in \autoref{exam:importorder}.
Similarly, for an example with multiple candidate stacks, we look at: \[\begin{tikzpicture}[remember picture, xscale=1.75]
    \node (p0) at (0,0) {\(f(\subnode{p00}{a},\subnode{p01}{b}){\downarrow}\)};
    \node (p1) at (0,1) {\(g(\subnode{p10}{a}){\downarrow}\phantom{,b}\)};
    \node (p2) at (0,2) {\(f(\subnode{p20}{a\hcancel{c}},\subnode{p21}{b}){\downarrow}\)};
    \path[->] (p00.north) edge node[left] {\(\rotatebox{90}{\(\ge\)}\)} (p10.south);
    \path[->, out=90, in=-45] (p01.north) edge node[right] {\(\mathrlap{\rotatebox{90}{\(\ge\)}}\)} (p10.south east);
    \path[->, green] (p10.north) edge node[left] {\(\rotatebox{90}{\(>\)}\)} (p20.south);
    \path[->, out=0, in=0, blue, dashed] (p2.east) edge node[right] {\(a\)} (p0.east);
  \end{tikzpicture}\]
  Here at \(g\) the candidate stack \(a\) is chosen over the stack \(b\) since \(a\) is earlier in the lexicographic order and therefore older than \(b\).
  After this, at the root a new stack is started, which reuses the name \(b\).

For a more complicated example, showing the difficulty of the problem in general, we consider the set \(\sf{Tree}(\bb N)\) of binary trees, whose elements are generated by \(\sf{leaf}(n)\) for \(n\in\bb N\) and \(\sf{node}(l,r)\) for \(l,r\in\sf{Tree}\).
This set has the subtree strict partial order \(<\), which is the structural order for the inductive definition, or more explicitly: the transitive relation generated by \(l<\sf{node}(l,r)\) and \(r<\sf{node}(l,r)\).
There is also a partial ordering on a subset of \(\sf{Tree}(\bb N)\) defined by \(t_0\sqsubseteq t_1\) iff the sequence of labels of \(t_0\) and \(t_1\), from left to right, is monotone.
Since the ordering on \(\bb N\) is decidable, meaning that we have a computable function \(({-}\le{-}):\bb N\times\bb N\to2\), we see that \(\sqsubseteq \) is also decidable: \begin{align*}
        ({-}&\sqsubseteq {-}):\sf{Tree}(\bb N)\times\sf{Tree}(\bb N)\to 2, \\
        (\sf{leaf}(n_0)&\sqsubseteq\makesize[l]{\sf{leaf}(n_1))}{\sf{node}(l_1,r_1))}\coloneqq (n_0\le n_1), \\
        (\sf{leaf}(n_0)&\sqsubseteq \sf{node}(l_1,r_1))\coloneqq(\sf{leaf}(n_0)\sqsubseteq l_1)\land(l_1\sqsubseteq r_1), \\
        (\sf{node}(l_0,r_0)&\sqsubseteq \makesize[l]{\sf{leaf}(n_1))}{\sf{node}(l_1,r_1))}\coloneqq (l_0\sqsubseteq r_0)\land(r_0\sqsubseteq \sf{leaf}(n_1)), \\
        (\sf{node}(l_0,r_0)&\sqsubseteq \sf{node}(l_1,r_1))\coloneqq(l_0\sqsubseteq  r_0)\land(r_0\sqsubseteq l_1)\land(l_1\sqsubseteq r_1).
    \end{align*}
    Moreover, each of the 7 recursive calls of this function satisfies at least one of the following size-change graphs:
    \begin{align*}
        \qquad\qquad\begin{tikzpicture}[yscale=0.8,remember picture]
        \node at (0,0) {\(\subnode{x0f}{x_0} \sqsubseteq \subnode{x1f}{x_1}\)};
        \node at (0,1) {\(\subnode{y0f}{y_0} \sqsubseteq \subnode{y1f}{y_1}\)};
        \draw[->,green] (x1f.north) -- node[left=8pt] {\(\rotatebox{135}{\(>\)}\)} (y0f.south);
        \draw[->,green] (x1f.north) -- node[right] {\(\mathrlap{\rotatebox{90}{\(>\)}}\)} (y1f.south);
    \end{tikzpicture}\qquad\qquad
    \begin{tikzpicture}[yscale=0.8,remember picture]
        \node at (0,0) {\(\subnode{x0h}{x_0} \sqsubseteq \subnode{x1h}{x_1}\)};
        \node at (0,1) {\(\subnode{y0h}{y_0} \sqsubseteq \subnode{y1h}{y_1}\)};
        \draw[->]       (x0h.north) -- node[left]  {\(\rotatebox{90}{\(\ge\)}\)} (y0h.south);
        \draw[->,green] (x1h.north) -- node[right] {\(\mathrlap{\rotatebox{90}{\(>\)}}\)} (y1h.south);
    \end{tikzpicture}\qquad\qquad
    \begin{tikzpicture}[yscale=0.8,remember picture]
        \node at (0,0) {\(\subnode{x0g}{x_0} \sqsubseteq \subnode{x1g}{x_1}\)};
        \node at (0,1) {\(\subnode{y0g}{y_0} \sqsubseteq \subnode{y1g}{y_1}\)};
        \draw[->,green] (x0g.north) -- node[left]  {\(\rotatebox{90}{\(>\)}\)} ( y0g.south);
        \draw[->]       (x1g.north) -- node[right] {\(\mathrlap{\rotatebox{90}{\(\ge\)}}\)} (y1g.south);
    \end{tikzpicture}\qquad\qquad
    \begin{tikzpicture}[yscale=0.8,remember picture]
        \node at (0,0) {\(\subnode{x0i}{x_0} \sqsubseteq \subnode{x1i}{x_1}\)};
        \node at (0,1) {\(\subnode{y0i}{y_0} \sqsubseteq \subnode{y1i}{y_1}\)};
        \draw[->,green] (x0i.north) -- node[left] {\(\rotatebox{90}{\(>\)}\)} (y0i.south);
        \draw[->,green] (x0i.north) -- node[right = 8pt] {\(\mathrlap{\rotatebox{45}{\(>\)}}\)} (y1i.south);
    \end{tikzpicture}
\end{align*}
So, we can fine a cyclic proof system in the style of \autoref{exam:run1}, and since this function satisfies the size-change termination principle, it has a unique infinite derivation in this cyclic proof system.
If we focus on the two size-change graphs on the right (since the two on the left are their mirror images), then we obtain the relaxed reset proof:
\[\begin{tikzpicture}[remember picture, yscale=1.3]
    \node (q0)    at ( 3,  0) {\((\subnode{q00}{a}\sqsubseteq\subnode{q01}{b}){\downarrow}\)};
    \node (q1)    at ( 2,  1) {\((\subnode{q10}{ac}\sqsubseteq\subnode{q11}{ad}){\downarrow}\)};
    \node (q1')   at (4.5, 1) {\((\subnode{q10'}{a\hcancel c}\sqsubseteq\subnode{q11'}{b}){\downarrow}\)};
    \node (q2)    at ( 0,  2) {\((\subnode{q20}{a\hcancel{cb}}\sqsubseteq\subnode{q21}{a\hcancel{ce}}){\downarrow}\)};
    \node (q2')   at ( 3,  2) {\((\subnode{q20'}{ac\hcancel b}\sqsubseteq\subnode{q21'}{ad}){\downarrow}\)};
    \node (q3)    at (-2.6,3) {\((\subnode{q30}{ab}\sqsubseteq\subnode{q31}{ac}){\downarrow}\)};
    \node (q3')   at ( 2.6,3) {\((\subnode{q30'}{ab}\sqsubseteq\subnode{q31'}{a}){\downarrow}\)};
    \node (q4)    at (-3.9,4) {\((\subnode{q40}{a\hcancel{bd}}\sqsubseteq\subnode{q41}{a\hcancel{be}}){\downarrow}\)};
    \node (q4')   at (-1.3,4) {\((\subnode{q40'}{ab\hcancel d}\sqsubseteq\subnode{q41'}{ac}){\downarrow}\)};
    \node (q4'')  at ( 1.3,4) {\((\subnode{q40''}{a\hcancel{bd}}\sqsubseteq\subnode{q41''}{a\hcancel{be}}){\downarrow}\)};
    \node (q4''') at ( 3.9,4) {\((\subnode{q40'''}{ab\hcancel d}\sqsubseteq\subnode{q41'''}{a}){\downarrow}\)};
    \path[->, green] (q00.north) edge (q10.south);
    \path[->, green] (q00.north) edge (q11.south);
    \path[->, green] (q00.north) edge (q10'.south);
    \path[->] (q01.north) edge (q11'.south);
    \path[->, green] (q10.north) edge (q20.south);
    \path[->, green] (q10.north) edge (q21.south);
    \path[->, green] (q10.north) edge (q20'.south);
    \path[->] (q11.north) edge (q21'.south);
    \path[->, green] (q20.north) edge (q30.south);
    \path[->, green] (q20.north) edge (q31.south);
    \path[->, green] (q20.north) edge (q30'.south);
    \path[->] (q21.north) edge (q31'.south);
    \path[->, green] (q30.north) edge (q40.south);
    \path[->, green] (q30.north) edge (q41.south);
    \path[->, green] (q30.north) edge (q40'.south);
    \path[->] (q31.north) edge (q41'.south);
    \path[->, green] (q30'.north) edge (q40''.south);
    \path[->, green] (q30'.north) edge (q41''.south);
    \path[->, green] (q30'.north) edge (q40'''.south);
    \path[->] (q31'.north) edge (q41'''.south);
    \path[->, out=-90, in=180, blue, dashed] (q4.south west) edge node[below left] {\(a\)} (q2.west);
    \path[->, out=-90, in=0, blue, dashed] (q4''.south west) edge node[right,pos=0.2] {\(a\)} (q2.east);
    \path[->, out=-90, in=0, red, dashed] (q4'.south east) edge node[below right] {\(b\)} (q3.east);
    \path[->, out=-90, in=0, red, dashed] (q4'''.south east) edge node[below right] {\(b\)} (q3'.east);
    \path[->, out=0, in=0, cyan, dashed] (q2') edge node[below right] {\(c\)} (q1);
    \path[->, out=0, in=0, blue, dashed] (q1') edge node[below right] {\(a\)} (q0);
\end{tikzpicture}\]

\end{document}